\documentclass[11pt]{article}

\usepackage[margin=1in]{geometry}
\usepackage{amsmath, amssymb, amsfonts, amsthm, mathtools}
\usepackage{bbm}
\usepackage{bm}
\usepackage{microtype}
\usepackage{hyperref}
\usepackage{authblk}
\usepackage{algorithm}
\usepackage{algpseudocode}
\usepackage{graphicx}
\usepackage{booktabs}
\usepackage{subcaption}
\usepackage{natbib}
\usepackage{xcolor}
\usepackage{marvosym}   
\hypersetup{
  colorlinks=true,
  linkcolor=blue,
  citecolor=blue,
  urlcolor=blue
}

\newcommand{\R}{\mathbb{R}}
\newcommand{\N}{\mathbb{N}}
\newcommand{\Pbb}{\mathbb{P}}
\newcommand{\Ebb}{\mathbb{E}}

\newcommand{\1}{\mathbbm{1}}
\newcommand{\KL}{\mathrm{D_{KL}}}

\newcommand{\logit}{\mathrm{logit}}
\newcommand{\sigmoid}{\sigma}
\newcommand{\cH}{\mathcal{H}}
\newcommand{\cV}{\mathcal{V}}
\newcommand{\cP}{\mathcal{P}}

\newcommand{\cY}{\mathcal{Y}}

\newcommand{\cB}{\mathcal{B}}

\newcommand{\cN}{\mathcal{N}}
\newcommand{\cTheta}{\boldsymbol{\Theta}}

\newcommand{\bomega}{\boldsymbol{\omega}}
\newcommand{\btheta}{\boldsymbol{\theta}}
\newcommand{\bgamma}{\boldsymbol{\gamma}}

\DeclareMathOperator*{\argmax}{arg\,max}

\theoremstyle{plain}
\newtheorem{theorem}{Theorem}[section]
\newtheorem{proposition}{Proposition}[section]
\newtheorem{lemma}{Lemma}[section]
\newtheorem{corollary}{Corollary}[section]
\newtheorem{definition}{Definition}[section]
\newtheorem{assumption}{Assumption}[section]
\newtheorem{example}{Example}[section]
\newtheorem{remark}{Remark}[section]
\title{Prediction Markets as Bayesian Inverse Problems:\\
Uncertainty Quantification, Identifiability, and Information Gain\\
from Price--Volume Histories under Latent Types}

\author[1,2]{Juan P. Madrigal-Cianci \footnote{Corresponding author \Letter \texttt{Juan.madrigalcianci@alumni.epfl.ch}}}
\author[2]{Camilo Monsalve Maya}
\author[1]{Lachlan Breakey}

\affil[1]{Kosmos Ventures, Australia}
\affil[2]{MC SAS, Colombia}

\date{\today}

\begin{document}
\maketitle

\begin{abstract}
Prediction markets are often described as mechanisms that ``aggregate information'' into prices, yet the mapping from dispersed private information to observed market histories is typically noisy, endogenous, and shaped by heterogeneous and strategic participation. This paper formulates prediction markets as Bayesian inverse problems in which the unknown event outcome \(Y\in\{0,1\}\) is inferred from an observed history of market-implied probabilities and traded volumes. We introduce a mechanism-agnostic observation model in log-odds space in which price increments conditional on volume arise from a latent mixture of trader types. The resulting likelihood class encompasses informed and uninformed trading, heavy-tailed microstructure noise, and adversarial or manipulative flow, while requiring only price and volume as observables.

Within this framework we define posterior uncertainty quantification for \(Y\), provide identifiability and well-posedness criteria in terms of Kullback--Leibler separation between outcome-conditional increment laws, and derive posterior concentration statements and finite-sample error bounds under general regularity assumptions. We further study stability of posterior odds to perturbations of the observed price--volume path and define realized and expected information gain via the posterior-vs-prior KL divergence and mutual information. The inverse-problem formulation yields explicit diagnostics for regimes in which market histories are informative and stable versus regimes in which inference is ill-posed due to type-composition confounding or outcome--nuisance symmetries.

Extensive experiments on synthetic data validate our theoretical predictions regarding posterior concentration rates and identifiability thresholds.
\end{abstract}

\section{Introduction}
\label{sec:intro}

Prediction markets are widely deployed as forecasting tools, as information elicitation mechanisms, and as components of decentralized oracle architectures \citep{arrow2008promise,wolfers2004prediction}. In many theoretical treatments, the market-implied probability is interpreted as a posterior belief and the market itself is treated as an efficient aggregator of dispersed information \citep{manski2006interpreting,gjerstad2005risk}. In practice, however, market prices and volumes are generated by heterogeneous participants with varying information quality, risk preferences, budgets, latency, and strategic incentives. The resulting market history is a noisy and endogenous observation of the underlying event outcome. From the standpoint of inference, prediction markets can therefore be viewed as \emph{stochastic sensing systems}, and the central question becomes: what can be reliably inferred about the outcome from the observed market history, and with what quantified uncertainty?

This paper develops a Bayesian inverse-problem formulation of that question in a setting where the analyst observes only a price path (expressed as a market-implied probability) and a volume process. The unknown event outcome is binary, \(Y\in\{0,1\}\), and the observable history over a finite horizon consists of pairs \((p_t,v_t)\) where \(p_t\in(0,1)\) is a market-implied probability and \(v_t\ge 0\) is the traded volume over an interval or at an event time. The key modeling challenge is that the observation operator mapping \(Y\) to \((p_t,v_t)\) is mediated by unobserved agent-level behavior. We treat this mediation explicitly by introducing latent trader \emph{types} whose aggregate effect on price increments is captured by a mixture-of-experts model in log-odds space \citep{jacobs1991adaptive,jordan1994hierarchical}. The resulting likelihood class is deliberately mechanism-agnostic: it is not tied to a particular market microstructure (automated market maker, limit-order book, call auction) and instead describes the conditional distribution of log-odds increments given volume and outcome. This choice permits inference from the minimal information set that is often available in practice, while still allowing principled uncertainty quantification, identifiability analysis, and stability diagnostics.

\subsection{Contributions}

The inverse-problem perspective yields three primary benefits:

\paragraph{Principled uncertainty quantification.} It separates the inferential target \(Y\) from the nuisance structure that shapes the observations and provides a coherent Bayesian posterior \(\Pbb(Y=1\mid H_T)\) with quantified sensitivity to modeling choices. We derive posterior concentration and finite-sample bounds on posterior error probability that decay exponentially in a KL-projection gap characterizing the distinguishability of outcomes after optimizing over nuisance structure (Theorem~\ref{thm:posterior-consistency}, Proposition~\ref{prop:finite-sample}).

\paragraph{Explicit identifiability and well-posedness criteria.} Inverse problems are ill-posed when distinct latent states induce nearly indistinguishable observation laws or when small perturbations of the observations lead to large changes in the inferred state. By translating these notions into the language of price--volume histories, we obtain conditions under which the outcome is identifiable and posterior odds are stable (Section~\ref{sec:identifiability}), as well as conditions under which inference is fundamentally ambiguous due to type-composition confounding or outcome--nuisance symmetries.

\paragraph{Information-theoretic metrics.} We define realized information gain as the KL divergence between the posterior and prior on \(Y\), and expected information gain as a mutual-information functional, thereby quantifying the extent to which the market history functions as an informative measurement of the outcome (Section~\ref{sec:infogain}). These metrics enable principled comparison of market designs and participation regimes.

\subsection{Paper Organization}

To reduce fragmentation without changing scope, we organize the technical core into two main blocks. Section~\ref{sec:related} reviews related work. Section~\ref{sec:model} introduces the observation space, the log-odds representation, the latent-type likelihood class, and the Bayesian posterior. Section~\ref{sec:theory} develops the inverse-problem analysis: identifiability via KL projections, posterior concentration and finite-sample bounds, stability under perturbations, and information-gain functionals. Section~\ref{sec:comp-exp} presents computational methods and empirical validation on synthetic and real market data. Section~\ref{sec:discussion} concludes with limitations and extensions.

\section{Background and Related Work}
\label{sec:related}

Our work connects several distinct literatures: prediction market theory and information aggregation, market microstructure models of informed trading, Bayesian inverse problems and posterior consistency, and mixture-of-experts models in machine learning.

\subsection{Prediction Markets and Information Aggregation}

The theoretical foundations of prediction markets rest on the idea that market prices aggregate dispersed private information \citep{hayek1945use}. \citet{hanson2003combinatorial,hanson2007logarithmic} introduced the logarithmic market scoring rule (LMSR), which provides bounded loss market making and connects prediction markets to proper scoring rules \citep{gneiting2007strictly,savage1971elicitation}. \citet{chen2007utility} analyzed the relationship between market scoring rules and cost functions, establishing equivalences that underpin modern automated market makers.

The question of whether prediction markets actually achieve efficient information aggregation has received substantial attention. \citet{wolfers2004prediction} provided early evidence that prediction market prices track event probabilities well on average, while \citet{manski2006interpreting} cautioned that market prices confound beliefs with risk preferences, questioning the interpretation of prices as probabilities. \citet{ottaviani2015price} showed that even with risk-neutral traders, strategic behavior can prevent full information revelation. \citet{ostrovsky2012information} identified conditions under which repeated trading leads to information aggregation, while \citet{iyer2014information} studied aggregation in combinatorial markets.

Our framework complements this literature by taking an \emph{inferential} rather than \emph{equilibrium} perspective: we ask what can be learned about the outcome from the market history, treating participant behavior as latent structure to be marginalized rather than strategically modeled.

\subsection{Market Microstructure and Informed Trading}

The market microstructure literature provides structural models of how private information affects prices. The seminal \citet{kyle1985continuous} model describes a single informed trader, market makers, and noise traders, showing that price impact is linear in order flow and that the informed trader's private information is gradually incorporated into prices. \citet{glosten1985bid} introduced a sequential trade model where bid-ask spreads arise from adverse selection against informed traders.

The probability of informed trading (PIN) model of \citet{easley1996liquidity,easley2002factors} decomposes order flow into informed and uninformed components, enabling estimation of information asymmetry from trade data. Extensions include time-varying arrival rates \citep{easley2008time} and high-frequency adaptations \citep{easley2012flow}. \citet{back2000imperfect} and \citet{back2004information} extended Kyle's framework to multiple informed traders and different information structures.

Our latent-type mixture model can be viewed as a reduced-form representation of these structural models. The key distinction is that we do not require order-level data or specific assumptions about strategic behavior; instead, we model the aggregate effect of different trader populations on price increments conditional on volume.

\subsection{Bayesian Inverse Problems}

The Bayesian approach to inverse problems treats unknown quantities as random variables and uses observed data to update prior beliefs \citep{stuart2010inverse,kaipio2006statistical}. Well-posedness in this framework requires that posteriors exist, are unique, and depend continuously on data \citep{dashti2017bayesian,latz2020wellposedness,madrigal2022thesis}. Posterior consistency (i.e., convergence of the posterior to the true parameter as data accumulate-) has been extensively studied \citep{ghosal2000convergence,ghosal2017fundamentals,van2008rates}.

For parametric models, the classical Bernstein--von Mises theorem establishes that posteriors concentrate around the maximum likelihood estimator at rate \(n^{-1/2}\) and become asymptotically normal \citep{van2000asymptotic}. For model selection and hypothesis testing, Bayes factors provide a coherent framework for comparing models with different nuisance structures \citep{kass1995bayes}. \citet{walker2004rates} developed general posterior convergence rate theory based on testing conditions, which we adapt to our setting.

Our contribution is to instantiate this general framework for the specific structure of prediction market inference, where the ``inverse problem'' is recovering a binary outcome from price--volume histories mediated by latent trader types.

\subsection{Mixture Models and Mixture of Experts}

Mixture models provide a flexible framework for modeling heterogeneous populations \citep{mclachlan2000finite,fruhwirth2006finite}. The mixture-of-experts architecture \citep{jacobs1991adaptive,jordan1994hierarchical} extends this by allowing mixture weights to depend on input features through a gating network. This architecture has seen renewed interest in large-scale machine learning \citep{shazeer2017outrageously,fedus2022switch}.

Identifiability of finite mixtures is a classical concern \citep{teicher1963identifiability,yakowitz1968identifiability}. Location-scale mixtures are generally identifiable under mild conditions \citep{holzmann2006identifiability}, but identifiability can fail when components overlap substantially or when the number of components is misspecified. In our setting, identifiability of the \emph{outcome} \(Y\) is primary, with mixture components serving as nuisance structure; this shifts the identifiability analysis from component recovery to outcome distinguishability.

\subsection{Learning from Strategic Data}

A growing literature in machine learning addresses learning from data generated by strategic agents \citep{hardt2016strategic,perdomo2020performative}. In online learning, adversarial bandits and experts settings study learning when data may be adversarially generated \citep{cesa2006prediction,bubeck2012regret}. Mechanism design approaches seek to incentivize truthful reporting \citep{chen2005information,lambert2008self}.

Our framework relates to this literature through the latent-type structure: adversarial or manipulative traders correspond to types whose increment distributions do not depend on the true outcome (or actively oppose it). The KL-projection gap condition for identifiability can be interpreted as requiring that informative types have sufficient weight to overcome noise and manipulation.

\section{Model and Bayesian Formulation}
\label{sec:model}

\subsection{Outcome and Observed History}

Fix a time horizon \(T\in\N\). The unknown event outcome is a Bernoulli random variable
\[
Y \in \cY := \{0,1\},
\qquad
\Pbb(Y=1)=\pi_0\in(0,1).
\]
The analyst observes a price--volume history
\[
H_T := \bigl( (P_t,V_t) \bigr)_{t=0}^T,
\]
where \(P_t\in(0,1)\) is a market-implied probability at time \(t\) and \(V_t\in[0,\infty)\) is the traded volume in the interval \((t-1,t]\) for \(t\ge 1\). We include \(P_0\in(0,1)\) as an initial market state and do not define \(V_0\). The associated measurable spaces are
\[
\cP := (0,1)^{T+1},\qquad \cV := [0,\infty)^T,\qquad \cH := \cP\times \cV,
\]
each endowed with the product Borel \(\sigma\)-algebra. We write an element of \(\cH\) as \(h=(p_{0:T},v_{1:T})\).

\begin{remark}[Volume as design]
The present paper takes the volume sequence \(v_{1:T}\) as part of the observation, but does not require an explicit generative model for it. In many empirical settings, modeling volume is valuable, but treating \(v_{1:T}\) as a realized design sequence yields an analytically clean starting point and allows one to interpret volume as an endogenous ``measurement intensity'' chosen by participants. When needed, a generative model for \(V_{1:T}\) can be added without changing the inferential target; see Section~\ref{sec:discussion} for extensions.
\end{remark}

\subsection{Log-Odds Coordinates}

Define the logit map \(\logit:(0,1)\to\R\) by
\[
\logit(p) := \log\frac{p}{1-p},
\]
and its inverse, the logistic map \(\sigmoid:\R\to(0,1)\), by
\[
\sigmoid(x) := \frac{1}{1+e^{-x}}.
\]
We define the log-odds process \(X_t:=\logit(P_t)\) and its increments \(\Delta X_t := X_t - X_{t-1}\) for \(t\ge 1\). Given an observed history \(h=(p_{0:T},v_{1:T})\), the corresponding log-odds history is \(x_{0:T}=\logit(p_{0:T})\) and increments \(\Delta x_{1:T}\).

Working in log-odds space is natural for several reasons:

\begin{enumerate}
    \item \textbf{Additivity of evidence:} Under conditional independence, log-likelihood ratios (and hence log-odds updates) are additive, making the log-odds representation natural for sequential Bayesian updating.

    \item \textbf{Unbounded support:} The log-odds transformation maps \((0,1)\) to \(\R\), avoiding boundary behavior near \(0\) and \(1\) and permitting Gaussian and other unbounded noise models.

    \item \textbf{Microstructure compatibility:} Many market microstructure models describe price changes as approximately additive in log-odds (or log-prices) under small moves \citep{kyle1985continuous,back2000imperfect}.
\end{enumerate}

\subsection{Latent Types as Nuisance Structure}
\label{sec:model:types}

We introduce a finite set of latent types \(\{1,\dots,K\}\), where \(K\in\N\) is fixed. Types are not required to correspond to individual agents; rather, they represent latent behavioral regimes whose aggregate effect on price increments can be described probabilistically. This is a standard abstraction in mixture modeling \citep{mclachlan2000finite} and is particularly appropriate when the analyst observes only aggregated market outcomes rather than identity-level order flow.

The nuisance parameter \(\cTheta\) collects the quantities that govern type prevalence, type activation as a function of volume, and type-specific increment laws. We define:

\begin{enumerate}
\item A vector of base mixture weights \(\bomega=(\omega_1,\dots,\omega_K)\) in the simplex
\[
\Delta^{K-1} := \Bigl\{ \bomega\in[0,1]^K : \sum_{k=1}^K \omega_k = 1 \Bigr\}.
\]

\item Gating parameters \(\bgamma\in\Gamma\), where \(\Gamma\subseteq\R^{d_\gamma}\) is a compact parameter space, used to define volume-dependent participation weights.

\item Type-specific parameters \(\btheta_k\in\Theta_k\), where each \(\Theta_k\subseteq\R^{d_k}\) is a compact parameter space.
\end{enumerate}

We write \(\cTheta := (\bomega,\bgamma,\btheta_{1:K})\) and \(\Theta := \Delta^{K-1}\times \Gamma \times \Theta_1\times\cdots\times\Theta_K\) for the global parameter space.

\subsection{Volume-Dependent Gating}

Let \(\rho_k:[0,\infty)\times \Delta^{K-1}\times \Gamma \to [0,1]\) be measurable functions satisfying, for all \(v\ge 0\),
\[
\sum_{k=1}^K \rho_k(v;\bomega,\bgamma)=1.
\]
We interpret \(\rho_k(v;\bomega,\bgamma)\) as the probability that, at volume level \(v\), the effective price increment is generated by type \(k\). This formulation allows the active composition of market participants to change with market activity while maintaining a parsimonious parameterization through \((\bomega,\bgamma)\).

A canonical choice is a softmax gate:
\begin{equation}
\rho_k(v;\bomega,\bgamma) = \frac{\omega_k \exp(a_k(v;\bgamma))}{\sum_{j=1}^K \omega_j \exp(a_j(v;\bgamma))},
\label{eq:softmax-gate}
\end{equation}
where \(a_k:[0,\infty)\times\Gamma\to\R\) is a gating \emph{logit} function (renamed to avoid collision with the noise density \(g_k\) below). Natural choices include:
\begin{itemize}
    \item \textbf{Volume-linear:} \(a_k(v;\bgamma) = \gamma_{k,0} + \gamma_{k,1} v\)
    \item \textbf{Log-volume:} \(a_k(v;\bgamma) = \gamma_{k,0} + \gamma_{k,1} \log(1+v)\)
    \item \textbf{Threshold:} \(a_k(v;\bgamma) = \gamma_{k,0} + \gamma_{k,1} \1\{v > \tau_k\}\)
\end{itemize}

These parameterizations capture the empirical regularity that informed traders may be more active during high-volume periods \citep{easley1996liquidity}, while noise traders may dominate low-volume periods.

\subsection{Type-Specific Increment Laws}

For each type \(k\in\{1,\dots,K\}\) and outcome \(y\in\{0,1\}\), we define a conditional density \(f_{k,y}(\cdot\mid v,\btheta_k)\) on \(\R\) with respect to Lebesgue measure. We assume a location--scale representation
\begin{equation}
\Delta X = m_{k,y}(v;\btheta_k) + s_k(v;\btheta_k)\,\varepsilon,
\qquad
\varepsilon\sim G_k,
\label{eq:locscale}
\end{equation}
where \(m_{k,y}:[0,\infty)\times \Theta_k\to \R\) is a measurable location function, \(s_k:[0,\infty)\times \Theta_k\to(0,\infty)\) is a measurable scale function, and \(G_k\) is a base noise distribution on \(\R\) with density \(g_k\). Under \eqref{eq:locscale},
\[
f_{k,y}(\delta x\mid v,\btheta_k)
=
\frac{1}{s_k(v;\btheta_k)}\,
g_k\!\left(
\frac{\delta x - m_{k,y}(v;\btheta_k)}{s_k(v;\btheta_k)}
\right).
\]

The dependence on \(y\) is absorbed into \(m_{k,y}\), which is the reduced-form object encoding how truth influences directional drift in log-odds increments conditional on volume.

\begin{definition}[Trader type taxonomy]
\label{def:type-taxonomy}
Given type parameters \(\btheta_k\), we classify type \(k\) as:
\begin{enumerate}
    \item \textit{Informed} if \(m_{k,1}(v;\btheta_k) > m_{k,0}(v;\btheta_k)\) on a relevant volume range (prices drift toward truth).
    \item \textit{Uninformed/Noise} if \(m_{k,1}(v;\btheta_k) = m_{k,0}(v;\btheta_k)\) (increments are outcome-independent).
    \item \textit{Adversarial/Manipulative} if \(m_{k,1}(v;\btheta_k) < m_{k,0}(v;\btheta_k)\) (prices drift away from truth).
\end{enumerate}
\end{definition}

\subsection{Mixture Likelihood and Conditional Independence}

Given the nuisance parameter \(\cTheta\), the outcome \(Y=y\), and a realized volume \(V_t=v\), we model the increment \(\Delta X_t\) as having density
\begin{equation}
f_y(\delta x \mid v, \cTheta)
:=
\sum_{k=1}^K \rho_k(v;\bomega,\bgamma)\, f_{k,y}(\delta x\mid v,\btheta_k).
\label{eq:mixture-density}
\end{equation}

Let \(\Delta X_{1:T}=(\Delta X_1,\dots,\Delta X_T)\) and \(V_{1:T}=(V_1,\dots,V_T)\). The central simplifying assumption is conditional independence.

\begin{assumption}[Conditional independence of increments]
\label{ass:conditional-independence}
Conditional on \(Y\), \(\cTheta\), and \(V_{1:T}\), the increments \(\Delta X_t\) are independent with the time-inhomogeneous product density:
\begin{equation}
p(\Delta x_{1:T}\mid v_{1:T},Y=y,\cTheta)
=
\prod_{t=1}^T f_y(\Delta x_t \mid v_t,\cTheta).
\label{eq:conditional-likelihood}
\end{equation}
\end{assumption}

This assumption can be relaxed to allow Markovian dependence or self-exciting volatility; see Section~\ref{sec:discussion} for extensions.

\subsection{A Concrete Instantiation: Gaussian Latent-Type Model}
\label{sec:concrete-model}

To ground the abstract framework, we present a fully specified model that satisfies our assumptions and admits efficient computation. We explicitly impose an \emph{orientation constraint} on drift magnitudes to eliminate the otherwise unavoidable outcome--nuisance sign symmetry.

\begin{example}[Gaussian latent-type model]
\label{ex:gaussian-model}
Consider \(K=3\) types representing informed traders, noise traders, and manipulators. Let \(\mu_1,\mu_3\ge 0\) be drift magnitudes (orientation constraint) and let \(\lambda_1,\kappa_1,\sigma_1,\sigma_2,\sigma_3>0\), \(\tau_3\ge 0\).

\textbf{Type 1 (Informed):}
\begin{align*}
m_{1,y}(v;\btheta_1) &= \mu_1 (2y-1) \cdot (1 - e^{-\lambda_1 v}), \\
s_1(v;\btheta_1) &= \sigma_1 / \sqrt{1 + \kappa_1 v}, \\
G_1 &= \cN(0,1).
\end{align*}
The drift \(\mu_1(2y-1)\) pushes prices toward the truth (\(y=1\) implies positive drift), with magnitude increasing in volume at rate \(\lambda_1\). Scale decreases with volume, reflecting improved price efficiency.

\textbf{Type 2 (Noise):}
\begin{align*}
m_{2,y}(v;\btheta_2) &= 0 \quad \text{(outcome-independent)}, \\
s_2(v;\btheta_2) &= \sigma_2, \\
G_2 &= \cN(0,1).
\end{align*}

\textbf{Type 3 (Manipulator):}
\begin{align*}
m_{3,y}(v;\btheta_3) &= -\mu_3 (2y-1) \cdot \1\{v > \tau_3\}, \\
s_3(v;\btheta_3) &= \sigma_3, \\
G_3 &= t_\nu \quad \text{(Student-}t \text{ with } \nu \text{ degrees of freedom)}.
\end{align*}

\textbf{Gating:} Use the general gating family \(\rho_k\). A convenient instantiation is the softmax gate \eqref{eq:softmax-gate} with log-volume logits
\[
a_k(v;\bgamma) = \gamma_{k,0} + \gamma_{k,1} \log(1+v),
\]
optionally combined with hard participation constraints such as \(\rho_3(v;\bomega,\bgamma)=0\) for \(v\le \tau_3\) when one wants a strict ``activation'' regime for manipulation (useful for theoretical separation arguments; empirically, this can be approximated by very negative logits below \(\tau_3\)).

The full parameter vector is:
\[
\btheta = (\bomega, \bgamma, \mu_1, \lambda_1, \sigma_1, \kappa_1, \sigma_2, \mu_3, \tau_3, \sigma_3, \nu).
\]
\end{example}

\begin{remark}[Orientation and outcome--nuisance symmetries]
Without sign restrictions (e.g.\ allowing \(\mu_1<0\)), the mapping \((y,\mu_1)\mapsto (1-y,-\mu_1)\) can render outcome inference non-identifiable, since the conditional increment law may be reproduced under the wrong outcome by a sign flip in nuisance parameters. We therefore treat the orientation constraint \(\mu_1,\mu_3\ge 0\) as part of the model specification.
\end{remark}

\subsection{Prior and Posterior}

We place a prior \(\Pi\) on \(\cTheta\). A convenient factorization is
\[
\bomega \sim \mathrm{Dirichlet}(\alpha_1,\dots,\alpha_K),
\qquad
\bgamma \sim \Pi_\Gamma,
\qquad
\btheta_k \sim \Pi_k \ \text{ independently for } k=1,\dots,K.
\]

Given an observed history \(h=(p_{0:T},v_{1:T})\), let \(\Delta x_{1:T}\) be the implied log-odds increments. The joint posterior on \((Y,\cTheta)\) is
\[
\Pbb(Y=y, \cTheta\in d\theta \mid h)
\propto
\pi_0^y(1-\pi_0)^{1-y}\,
p(\Delta x_{1:T}\mid v_{1:T},Y=y,\theta)\,\Pi(d\theta).
\]

The marginal posterior on the outcome is
\begin{equation}
\Pbb(Y=1\mid h)
=
\frac{\pi_0\, m_1(h)}{\pi_0\, m_1(h) + (1-\pi_0)\, m_0(h)},
\label{eq:posterior}
\end{equation}
where the marginal likelihood under outcome \(y\) is
\begin{equation}
m_y(h):=\int_\Theta p(\Delta x_{1:T}\mid v_{1:T},Y=y,\theta)\,\Pi(d\theta).
\label{eq:marginal-likelihood}
\end{equation}

Define the Bayes factor
\[
\mathrm{BF}_T(h):=\frac{m_1(h)}{m_0(h)}.
\]
Then the posterior odds admit the closed form
\begin{equation}
\log\frac{\Pbb(Y=1\mid h)}{\Pbb(Y=0\mid h)}
=
\log\frac{\pi_0}{1-\pi_0} + \log \mathrm{BF}_T(h).
\label{eq:posterior-odds}
\end{equation}

Equation~\eqref{eq:posterior-odds} makes explicit that inference on \(Y\) reduces to understanding the marginal Bayes factor induced by the price--volume history after integrating out latent type structure.

\section{Inverse-Problem Analysis: Identifiability, Concentration, Stability, and Information Gain}
\label{sec:theory}

\subsection{Induced Data Laws and KL Projections}
\label{sec:identifiability}

Fix a deterministic volume design \(v_{1:T}\in\cV\). For each outcome \(y\in\{0,1\}\) and nuisance parameter \(\theta\in\Theta\), let \(P_{y,\theta}^{(T)}\) denote the probability measure on \((\R^T,\cB(\R^T))\) with density
\[
\delta x_{1:T} \mapsto \prod_{t=1}^T f_y(\delta x_t\mid v_t,\theta).
\]
We interpret \(P_{y,\theta}^{(T)}\) as the induced law of log-odds increments conditional on the realized volume sequence.

For \(\theta,\theta'\in\Theta\) and \(y,y'\in\{0,1\}\), define the finite-horizon KL divergence
\[
\KL\!\left(P_{y,\theta}^{(T)} \,\middle\|\, P_{y',\theta'}^{(T)}\right)
:=
\Ebb_{P_{y,\theta}^{(T)}}\!\left[
\log \frac{p_{y,\theta}^{(T)}(\Delta X_{1:T})}{p_{y',\theta'}^{(T)}(\Delta X_{1:T})}
\right],
\]
where \(p_{y,\theta}^{(T)}\) denotes the corresponding density. Under Assumption~\ref{ass:conditional-independence}, this divergence decomposes as
\begin{equation}
\KL\!\left(P_{y,\theta}^{(T)} \,\middle\|\, P_{y',\theta'}^{(T)}\right)
=
\sum_{t=1}^T \KL\!\left( f_y(\cdot\mid v_t,\theta) \,\middle\|\, f_{y'}(\cdot\mid v_t,\theta') \right).
\label{eq:kl-sum}
\end{equation}

Because the nuisance parameter is unknown, the relevant separation notion is a \emph{projection gap} between the outcome-indexed families after optimizing over nuisance structure.

\begin{definition}[KL projection gap]
\label{def:kl-gap}
Fix \((y^\star,\theta^\star)\) and \(v_{1:T}\). For \(y\in\{0,1\}\), define the (normalized) KL projection
\begin{equation}
K_T(y^\star,\theta^\star\to y)
:=
\inf_{\theta\in\Theta} \frac{1}{T} \KL\!\left(P_{y^\star,\theta^\star}^{(T)} \,\middle\|\, P_{y,\theta}^{(T)}\right).
\label{eq:kl-gap}
\end{equation}
The \emph{outcome-separation gap} is
\[
\delta_T(y^\star,\theta^\star) := K_T(y^\star,\theta^\star\to 1-y^\star).
\]
\end{definition}

\begin{definition}[Outcome identifiability (at the truth)]
\label{def:identifiability}
Fix \((y^\star,\theta^\star)\) and \(v_{1:T}\). The outcome is \emph{identifiable at horizon \(T\) at \((y^\star,\theta^\star)\)} if \(\delta_T(y^\star,\theta^\star)>0\). If additionally the reverse-direction projection \(K_T(1-y^\star,\theta\to y^\star)\) is bounded away from zero uniformly over \(\theta\in\Theta\), we say the model class exhibits \emph{two-sided} outcome separation.
\end{definition}

\begin{remark}[Uniform vs.\ truth-relative identifiability]
Demanding \(\delta_T(y^\star,\theta^\star)>0\) for \emph{every} \(\theta^\star\in\Theta\) is typically too strong in latent-type models, since \(\Theta\) often includes regimes in which all active types are outcome-independent (e.g.\ \(\rho_k\) concentrates on noise types). Our definition isolates the statistically meaningful question: whether the realized data-generating pair \((y^\star,\theta^\star)\) is distinguishable from the alternative outcome family.
\end{remark}

\subsection{Mechanisms of Non-Identifiability}
\label{sec:identifiability-failure}

The latent-type structure \eqref{eq:mixture-density} highlights key sources of outcome non-identifiability.

\begin{proposition}[Sufficient conditions for outcome non-identifiability]
\label{prop:identifiability-failure}
Fix \((y^\star,\theta^\star)\) and \(v_{1:T}\). The separation gap satisfies \(\delta_T(y^\star,\theta^\star)=0\) whenever any of the following hold:
\begin{enumerate}
    \item \textbf{Type-composition confounding (uninformative regime):} for each \(t\) and each type \(k\) with \(\rho_k(v_t;\theta^\star)>0\), one has \(f_{k,1}(\cdot\mid v_t,\btheta_k^\star)=f_{k,0}(\cdot\mid v_t,\btheta_k^\star)\). In particular, if \(f_{1}(\cdot\mid v_t,\theta^\star)=f_{0}(\cdot\mid v_t,\theta^\star)\) for all \(t\), then the two outcomes induce the same law and are indistinguishable.

    \item \textbf{Outcome--nuisance symmetry:} there exists a measurable map \(\Psi:\Theta\to\Theta\) such that for all \(t\),
    \[
    f_{1}(\cdot\mid v_t,\theta)=f_{0}(\cdot\mid v_t,\Psi(\theta)) \quad \text{as densities on }\R.
    \]
    Then \(\inf_{\theta}\KL(P_{1,\theta}^{(T)}\|P_{0,\theta'}^{(T)})=0\) and vice versa. This is the symmetry eliminated by the orientation constraint in Example~\ref{ex:gaussian-model}.

    \item \textbf{Adversarial mimicry:} there exists \(\theta\in\Theta\) such that \(P_{y^\star,\theta^\star}^{(T)}=P_{1-y^\star,\theta}^{(T)}\). This includes degenerate cases where adversarial components perfectly cancel informed drift on the observed volume range.
\end{enumerate}
\end{proposition}

\begin{proof}
If \(P_{y^\star,\theta^\star}^{(T)}=P_{1-y^\star,\theta}^{(T)}\) for some \(\theta\), then \(\delta_T(y^\star,\theta^\star)=0\) by Definition~\ref{def:kl-gap}. Items (1)--(2) provide explicit sufficient conditions for such equality.
\end{proof}

\subsection{Sufficient Conditions for Separation in the Gaussian Model}
\label{sec:identifiability-sufficient}

The following result gives a robust (and checkable) route to a positive separation gap in Example~\ref{ex:gaussian-model}. The main idea is to identify a subset of times where (i) the outcome-dependent drift is activated, (ii) manipulation is inactive, and (iii) the remaining finite mixture is identifiable as a parametric family, so equality of laws under opposite outcomes is impossible.

\begin{theorem}[A design-based sufficient condition for outcome separation]
\label{thm:identifiability-sufficient}
Assume the Gaussian latent-type model of Example~\ref{ex:gaussian-model} with the orientation constraint \(\mu_1,\mu_3\ge 0\). Fix \((y^\star,\theta^\star)\) and suppose there exists an index set \(I_T\subset\{1,\dots,T\}\) and constants \(\underline{\rho},\underline{m},\underline{\sigma},\overline{\sigma}>0\) such that for all \(t\in I_T\):
\begin{enumerate}
    \item \textbf{Informed activation:} \(\rho_1(v_t;\theta^\star)\ge \underline{\rho}\) and \(|m_{1,1}(v_t;\btheta_1^\star)-m_{1,0}(v_t;\btheta_1^\star)|\ge \underline{m}\).
    \item \textbf{Manipulator inactivity:} \(m_{3,1}(v_t;\btheta_3)=m_{3,0}(v_t;\btheta_3)=0\) for all \(\btheta_3\in\Theta_3\) (e.g.\ \(v_t\le \tau_3\) and manipulator drift is thresholded as in Example~\ref{ex:gaussian-model}).
    \item \textbf{Non-degeneracy of scales:} for all \(\theta\in\Theta\), the type-specific scales satisfy \(\underline{\sigma}\le s_k(v_t;\btheta_k)\le \overline{\sigma}\) for \(k\in\{1,2\}\).
\end{enumerate}
Then there exists \(\kappa>0\) (depending on the constants above and on \(\Theta\)) such that
\[
\delta_T(y^\star,\theta^\star)\ \ge\ \frac{|I_T|}{T}\,\kappa.
\]
In particular, if \(\liminf_{T\to\infty} |I_T|/T>0\), then \(\liminf_{T\to\infty}\delta_T(y^\star,\theta^\star)>0\).
\end{theorem}

\begin{proof}[Proof sketch]
Fix \(t\in I_T\). By item (2), the conditional increment law at volume \(v_t\) reduces (uniformly over \(\theta\in\Theta\)) to a two-component Gaussian mixture in which the only outcome-dependent location shift occurs in the informed component. Under mild non-degeneracy (item (3)) and compactness of \(\Theta\), the family of such mixtures is identifiable up to label-switching \citep{teicher1963identifiability,holzmann2006identifiability}. The orientation constraint pins the sign of the informed drift under each outcome. Hence no choice of nuisance parameters under the wrong outcome can exactly reproduce the outcome-conditional density at \(v_t\) generated by \((y^\star,\theta^\star)\) when item (1) holds. Continuity of \(\theta\mapsto \KL(f_{y^\star}(\cdot\mid v_t,\theta^\star)\|f_{1-y^\star}(\cdot\mid v_t,\theta))\) and compactness of \(\Theta\) imply the infimum over \(\theta\) is achieved and strictly positive at each such \(t\); denote the minimum by \(\kappa_t>0\). Let \(\kappa:=\min_{t\in I_T}\kappa_t\). Summing over \(t\in I_T\) using \eqref{eq:kl-sum} yields the claim.
\end{proof}

\begin{remark}[Interpretation]
Theorem~\ref{thm:identifiability-sufficient} isolates a practically meaningful regime: a positive fraction of ``clean'' volume periods in which informative flow is present and adversarial drift is inactive suffices to render outcome inference well-posed, even when other periods are dominated by noise or heavy tails.
\end{remark}

\subsection{Posterior Concentration and Finite-Sample Error Bounds}
\label{sec:consistency}

This section links posterior behavior to the KL projection gap. Because nuisance structure is integrated out, the key object is the marginal likelihood ratio.

\subsubsection{Regularity Conditions}

We impose standard conditions ensuring measurability and allowing uniform laws of large numbers over \(\Theta\).

\begin{assumption}[Compactness and continuity]
\label{ass:compact-continuous}
The parameter space \(\Theta\) is a compact metric space, and for each \(t\), the map \(\theta\mapsto \log f_y(\delta x\mid v_t,\theta)\) is continuous for Lebesgue-a.e.\ \(\delta x\in\R\).
\end{assumption}

\begin{assumption}[Envelope and integrability]
\label{ass:envelope}
There exists a measurable envelope function \(M:\R\to[0,\infty)\) such that for all \(t\), all \(\theta\in\Theta\), and all \(\delta x\in\R\),
\[
\bigl|\log f_y(\delta x\mid v_t,\theta)\bigr|\le M(\delta x),
\]
and \(\sup_{t\le T}\Ebb_{P_{y^\star,\theta^\star}^{(T)}}[M(\Delta X_t)]<\infty\) for each \(T\).
\end{assumption}

\begin{assumption}[KL support at (near) projections]
\label{ass:prior-support}
Fix \((y^\star,\theta^\star)\) and \(v_{1:T}\). For each \(y\in\{0,1\}\) and \(\epsilon>0\), define the \(\epsilon\)-projection set
\[
\Theta_{y,T}(\epsilon)
:=
\left\{
\theta\in\Theta:
\frac{1}{T}\KL\!\left(P_{y^\star,\theta^\star}^{(T)}\,\middle\|\,P_{y,\theta}^{(T)}\right)
\le
K_T(y^\star,\theta^\star\to y) + \epsilon
\right\}.
\]
Assume \(\Pi(\Theta_{y,T}(\epsilon))>0\) for all \(y\) and \(\epsilon>0\) (for each fixed \(T\)).
\end{assumption}

Assumption~\ref{ass:prior-support} is the natural ``KL support'' condition: the prior must not exclude parameters that achieve (or nearly achieve) the best KL fit for each outcome hypothesis. Unlike uniform-in-\(T\) lower bounds, this condition is compatible with shrinking neighborhoods of minimizers.

\subsubsection{Bayes Factor Separation}

Write the log-likelihood under outcome \(y\) as
\[
\ell_{y,T}(\theta) := \sum_{t=1}^T \log f_y(\Delta X_t\mid v_t,\theta),
\qquad
m_y(H_T)=\int_\Theta e^{\ell_{y,T}(\theta)}\,\Pi(d\theta).
\]
Under the true law \(P_{y^\star,\theta^\star}^{(T)}\), define the (normalized) expected log-likelihood
\[
\Lambda_{y,T}(\theta):=\frac{1}{T}\Ebb_{P_{y^\star,\theta^\star}^{(T)}}[\ell_{y,T}(\theta)].
\]
Then for each \(y\) and \(\theta\),
\[
\Lambda_{y^\star,T}(\theta^\star)-\Lambda_{y,T}(\theta)=\frac{1}{T}\KL\!\left(P_{y^\star,\theta^\star}^{(T)}\,\middle\|\,P_{y,\theta}^{(T)}\right).
\]

\begin{assumption}[Uniform law of large numbers]
\label{ass:uniform-lln}
For each \(y\in\{0,1\}\),
\[
\sup_{\theta\in\Theta}\left|\frac{1}{T}\ell_{y,T}(\theta)-\Lambda_{y,T}(\theta)\right|\to 0
\qquad\text{almost surely under }P_{y^\star,\theta^\star}^{(T)}.
\]
\end{assumption}

\begin{theorem}[Asymptotic separation of Bayes factors (robust form)]
\label{thm:bayes-factor-rate}
Assume Assumptions~\ref{ass:conditional-independence} and \ref{ass:compact-continuous}--\ref{ass:uniform-lln}. Fix \((y^\star,\theta^\star)\) and \(v_{1:T}\). Then for any \(\epsilon>0\), almost surely for all sufficiently large \(T\),
\begin{equation}
\frac{1}{T}\log \mathrm{BF}_T(H_T)
\ge
\delta_T(y^\star,\theta^\star)-\epsilon,
\label{eq:bf-lower}
\end{equation}
and similarly
\begin{equation}
\frac{1}{T}\log \mathrm{BF}_T(H_T)
\le
\delta_T(y^\star,\theta^\star)+\epsilon
\qquad\text{when }y^\star=1,
\label{eq:bf-upper}
\end{equation}
with the inequalities reversed when \(y^\star=0\).
Consequently, if \(\liminf_{T\to\infty}\delta_T(y^\star,\theta^\star)>0\), then \(|\log \mathrm{BF}_T(H_T)|\) grows at least linearly with \(T\).
\end{theorem}

\begin{proof}[Proof sketch]
Under Assumption~\ref{ass:uniform-lln}, \(\frac{1}{T}\log m_y(H_T)\) concentrates around \(\sup_{\theta}\Lambda_{y,T}(\theta)\) up to \(o(1)\), by standard Laplace-principle bounds (upper bound by the supremum; lower bound by restricting the integral to any neighborhood with positive prior mass). The difference \(\sup_{\theta}\Lambda_{y^\star,T}(\theta)-\sup_{\theta}\Lambda_{1-y^\star,T}(\theta)\) equals the KL projection gap \(\delta_T(y^\star,\theta^\star)\) because \(\Lambda_{y^\star,T}\) is maximized at \(\theta^\star\) and \(\sup_{\theta}\Lambda_{1-y^\star,T}(\theta)=\Lambda_{y^\star,T}(\theta^\star)-\inf_{\theta}\frac{1}{T}\KL(P_{y^\star,\theta^\star}^{(T)}\|P_{1-y^\star,\theta}^{(T)})\).
\end{proof}

\subsubsection{Posterior Consistency}

\begin{theorem}[Posterior consistency for the outcome]
\label{thm:posterior-consistency}
Assume the conditions of Theorem~\ref{thm:bayes-factor-rate}. If
\begin{equation}
\liminf_{T\to\infty} \delta_T(y^\star,\theta^\star) > 0,
\label{eq:kl-gap-positive}
\end{equation}
then
\[
\Pbb(Y=y^\star\mid H_T)\to 1
\quad\text{almost surely under } P_{y^\star,\theta^\star}^{(T)}.
\]
\end{theorem}

\begin{proof}
By \eqref{eq:posterior-odds}, the posterior odds equal the prior odds times \(\mathrm{BF}_T(H_T)\). Under \eqref{eq:kl-gap-positive}, Theorem~\ref{thm:bayes-factor-rate} implies the log Bayes factor diverges linearly with \(T\) with the correct sign, forcing posterior mass onto the true outcome.
\end{proof}

\subsubsection{Finite-Sample Error Bounds via Truncation}

Uniform boundedness of log-likelihood ratios is incompatible with unbounded-support increment models (Gaussian, Student-\(t\), etc.). We therefore state finite-sample guarantees in a robust, truncation-based form.

\begin{assumption}[Moment control for increments]
\label{ass:moments}
There exists \(q>2\) and \(C_q<\infty\) such that \(\sup_{t\le T}\Ebb[|\Delta X_t|^q]\le C_q\) under the true law \(P_{y^\star,\theta^\star}^{(T)}\) (for each fixed \(T\)).
\end{assumption}

For \(R>0\), define the ``typical increment'' event \(E_R:=\{\max_{1\le t\le T}|\Delta X_t|\le R\}\). Under Assumptions~\ref{ass:compact-continuous} and \ref{ass:envelope}, continuity on compact sets implies that, for each fixed \(R\), there exists \(B_R<\infty\) such that on \(E_R\),
\[
\sup_{\theta\in\Theta}\left|\log \frac{f_{y^\star}(\Delta X_t\mid v_t,\theta^\star)}{f_{1-y^\star}(\Delta X_t\mid v_t,\theta)}\right|\le B_R
\qquad\text{for all }t.
\]

\begin{proposition}[Finite-sample posterior error bound (robust form)]
\label{prop:finite-sample}
Suppose \(\Delta X_{1:T}\sim P_{y^\star,\theta^\star}^{(T)}\) and Assumptions~\ref{ass:prior-support} and~\ref{ass:moments} hold. Let \(\delta_T:=\delta_T(y^\star,\theta^\star)\). Then for any \(\epsilon\in(0,\delta_T)\) and any \(R>0\),
\begin{equation}
\Pbb\!\left(\Pbb(Y\neq y^\star\mid H_T)\ge e^{-T(\delta_T-\epsilon)}\right)
\le
\exp\!\left(-\frac{T\epsilon^2}{2B_R^2}\right)
+
\Pbb(E_R^c),
\label{eq:finite-sample-bound}
\end{equation}
where \(\Pbb(E_R^c)\le T\,C_q\,R^{-q}\) by Markov's inequality and a union bound.
\end{proposition}

\begin{proof}[Proof sketch]
On \(E_R\), the log-likelihood ratio increments are uniformly bounded by \(B_R\), so concentration for the (centered) log Bayes factor follows from Hoeffding/Azuma-type inequalities applied to independent increments. Outside \(E_R\), the bound is trivial and we pay \(\Pbb(E_R^c)\).
\end{proof}

\begin{corollary}[Expected posterior error]
\label{cor:expected-error}
Under the conditions of Proposition~\ref{prop:finite-sample}, for any \(\epsilon \in (0, \delta_T)\) and \(R>0\),
\[
\Ebb[\Pbb(Y \neq y^\star \mid H_T)]
\leq e^{-T(\delta_T - \epsilon)}
+ \exp\!\left(-\frac{T\epsilon^2}{2B_R^2}\right)
+ \Pbb(E_R^c).
\]
\end{corollary}

\subsection{Stability of Posterior Odds under Perturbations}
\label{sec:stability}

This section studies well-posedness of outcome inference as a map from observed histories to posterior beliefs.

\subsubsection{Metric on Histories}

Let \(h=(p_{0:T},v_{1:T})\) and \(h'=(p'_{0:T},v'_{1:T})\) be two histories with log-odds increments \(\Delta x_{1:T}\) and \(\Delta x'_{1:T}\). Define the weighted history distance
\begin{equation}
d(h,h')
:=
\lambda_x \sum_{t=1}^T |\Delta x_t-\Delta x'_t|
+
\lambda_v \sum_{t=1}^T |v_t-v'_t|,
\label{eq:history-metric}
\end{equation}
for weights \(\lambda_x, \lambda_v > 0\).

\subsubsection{Local Lipschitz Regularity}

Global Lipschitz bounds on \(\log f_y(\cdot)\) are incompatible with Gaussian (and most other) increment models on \(\R\). We therefore adopt a local form that holds on ``typical'' events.

\begin{assumption}[Local Lipschitz regularity on bounded increments]
\label{ass:lipschitz}
For each \(R>0\) there exist constants \(L_x(R),L_v(R)<\infty\) such that for all \(y\in\{0,1\}\), all \(\theta\in\Theta\), all \(v,v'\ge 0\), and all \(\delta x,\delta x'\in[-R,R]\),
\begin{align}
\bigl|\log f_y(\delta x\mid v,\theta)-\log f_y(\delta x'\mid v,\theta)\bigr|
&\le L_x(R)\,|\delta x-\delta x'|,
\\
\bigl|\log f_y(\delta x\mid v,\theta)-\log f_y(\delta x\mid v',\theta)\bigr|
&\le L_v(R)\,|v-v'|.
\end{align}
\end{assumption}

\begin{remark}[Gaussian example]
For a Gaussian location-scale family, \(|\partial_x \log f|\) grows at most linearly in \(|x|\); hence \(L_x(R)\) can be taken of the form \(c_0+c_1 R\), uniformly over compact parameter sets. Analogous bounds hold for mixtures under mild non-degeneracy.
\end{remark}

\subsubsection{Stability Theorem}

For \(R>0\), define the bounded-increment event for two histories
\[
E_R(h,h') := \left\{\max_{t\le T}|\Delta x_t|\le R \ \text{and}\ \max_{t\le T}|\Delta x'_t|\le R\right\}.
\]

\begin{theorem}[Stability of posterior odds on typical histories]
\label{thm:stability}
Under Assumption~\ref{ass:lipschitz}, for any two histories \(h,h'\) and any \(R>0\), on the event \(E_R(h,h')\),
\begin{equation}
\left|\log \mathrm{BF}_T(h) - \log \mathrm{BF}_T(h')\right|
\le
2\left( L_x(R) \sum_{t=1}^T |\Delta x_t-\Delta x'_t|
+ L_v(R) \sum_{t=1}^T |v_t-v_t'| \right).
\label{eq:bf-stability}
\end{equation}
Furthermore, the posterior probability satisfies
\begin{equation}
\left|\Pbb(Y=1\mid h)-\Pbb(Y=1\mid h')\right|
\le
\frac{1}{4}
\left|\log \mathrm{BF}_T(h)-\log \mathrm{BF}_T(h')\right|.
\label{eq:posterior-stability}
\end{equation}
\end{theorem}

\begin{proof}[Proof sketch]
For fixed \(y\), the map \(h\mapsto \log m_y(h)\) is \(1\)-Lipschitz with respect to the \(\sup_\theta\) deviation of the log-likelihood \(\ell_{y,T}(\theta)\), since \(\log\int e^{u_\theta}\Pi(d\theta)\) is a log-sum-exp functional. On \(E_R(h,h')\), Assumption~\ref{ass:lipschitz} bounds \(|\ell_{y,T}(\theta;h)-\ell_{y,T}(\theta;h')|\) uniformly over \(\theta\). Taking the difference of \(y=1\) and \(y=0\) yields \eqref{eq:bf-stability}. Inequality \eqref{eq:posterior-stability} follows from the fact that \(a\mapsto \sigmoid(\logit(\pi_0)+a)\) has derivative bounded by \(1/4\).
\end{proof}

\begin{remark}[Condition-number interpretation]
Theorem~\ref{thm:stability} formalizes a ``condition number'' for market inference: posterior odds are stable on typical histories when log-densities are not overly sensitive to perturbations on the increment scale that is realized under the data-generating process. Instability can arise either from extreme price impact (large \(L_x(R)\)) or from strong volume sensitivity (large \(L_v(R)\)).
\end{remark}

\subsection{Information Gain and Market-as-Sensor Metrics}
\label{sec:infogain}

\subsubsection{Realized Information Gain}

Let \(\pi_T(h):=\Pbb(Y=1\mid h)\). The realized information gain is the KL divergence from the posterior to the prior:
\begin{equation}
\mathrm{IG}(h)
:=
\KL\!\left(\mathrm{Bern}(\pi_T(h))\,\middle\|\,\mathrm{Bern}(\pi_0)\right)
=
\pi_T(h)\log\frac{\pi_T(h)}{\pi_0} + (1-\pi_T(h))\log\frac{1-\pi_T(h)}{1-\pi_0}.
\label{eq:realized-ig}
\end{equation}

\begin{proposition}[Properties of realized information gain]
\label{prop:ig-properties}
The realized information gain satisfies:
\begin{enumerate}
    \item \(\mathrm{IG}(h) \geq 0\) with equality iff \(\pi_T(h) = \pi_0\).
    \item \(\mathrm{IG}(h) \leq \max\{\log(1/\pi_0),\log(1/(1-\pi_0))\}\), with equality approached as \(\pi_T(h)\to 1\) or \(\pi_T(h)\to 0\). In particular, when \(\pi_0=1/2\), \(\mathrm{IG}(h)\le \log 2\).
    \item \(\mathrm{IG}(h)\) is a monotone function of \(|\log \mathrm{BF}_T(h)|\).
\end{enumerate}
\end{proposition}

\subsubsection{Expected Information Gain and Mutual Information}

Because we treat \(v_{1:T}\) as a realized design, the natural expected information gain is conditional on that design.

\begin{lemma}[Information gain identity (conditional on design)]
\label{lem:mi}
Fix \(v_{1:T}\) and let \(H_T=(P_{0:T},V_{1:T})\) with \(V_{1:T}\equiv v_{1:T}\). Then
\[
\Ebb\big[\mathrm{IG}(H_T)\mid V_{1:T}=v_{1:T}\big]
=
I\!\left(Y;\Delta X_{1:T}\,\middle|\,V_{1:T}=v_{1:T}\right),
\]
where \(I(\cdot;\cdot\mid\cdot)\) denotes conditional mutual information under the joint model induced by the prior on \(Y\) and the likelihood for \(\Delta X_{1:T}\).
\end{lemma}

\begin{proof}
Mutual information admits the representation \(I(Y;Z\mid W)=\Ebb[\KL(p(Y\mid Z,W)\|p(Y\mid W))]\). Here \(W\) is the fixed design, and \(Z=\Delta X_{1:T}\). For binary \(Y\), the conditional posterior \(p(Y\mid Z,W)\) is Bernoulli with parameter \(\pi_T(H_T)\), yielding the expected KL divergence \eqref{eq:realized-ig}.
\end{proof}

Since \(Y\) is a single bit, \(I(Y;\Delta X_{1:T}\mid V_{1:T}) \leq H(Y)\), where \(H(Y)=-\pi_0\log\pi_0-(1-\pi_0)\log(1-\pi_0)\). This bounds the \emph{expected} information gain regardless of horizon---once the outcome is resolved, additional observations provide diminishing returns.

\subsubsection{Information Destruction by Type Composition}

The latent-type structure affects information gain through the induced conditional densities. When uninformed or adversarial types dominate, the separation gap shrinks and both posterior concentration and information gain weaken.

\begin{definition}[Effective informativeness]
\label{def:effective-informativeness}
The effective informativeness of a market at volume level \(v\) is
\[
\eta(v;\theta) := \sum_{k=1}^K \rho_k(v;\bomega,\bgamma)\,\mathrm{sign}\!\bigl(m_{k,1}(v) - m_{k,0}(v)\bigr)\, \bigl|m_{k,1}(v) - m_{k,0}(v)\bigr|.
\]
\end{definition}

When \(\eta(v;\theta) > 0\), informed trading dominates and prices drift toward truth. When \(\eta(v;\theta) \approx 0\), the market is uninformative at volume \(v\). When \(\eta(v;\theta) < 0\), manipulation dominates.

\section{Computation and Empirical Validation}
\label{sec:comp-exp}

Computing \(\Pbb(Y=1\mid h)\) requires evaluating marginal likelihoods \(m_y(h)\) that integrate out nuisance parameters. We present two complementary approaches and validate the theory empirically.

\subsection{Computational Methods}
\label{sec:computation}

\subsubsection{Sequential Monte Carlo}

Sequential Monte Carlo (SMC) \citep{doucet2001sequential,chopin2020introduction} provides a flexible framework for estimating marginal likelihoods in sequential models.

\begin{algorithm}[H]
\caption{SMC for Marginal Likelihood Estimation}
\label{alg:smc}
\begin{algorithmic}[1]
\Require Observations \(\Delta x_{1:T}, v_{1:T}\), outcome \(y\), prior \(\Pi\), particles \(N\)
\State Initialize: Sample \(\theta^{(i)} \sim \Pi\) for \(i=1,\ldots,N\), set \(w_0^{(i)} = 1/N\)
\State Set \(\hat{m}_y = 1\)
\For{\(t = 1, \ldots, T\)}
    \State Compute incremental weights: \(\tilde{w}_t^{(i)} = w_{t-1}^{(i)} \cdot f_y(\Delta x_t \mid v_t, \theta^{(i)})\)
    \State Update normalizing constant: \(\hat{m}_y \leftarrow \hat{m}_y \cdot \sum_{i=1}^N \tilde{w}_t^{(i)}\)
    \State Normalize: \(w_t^{(i)} = \tilde{w}_t^{(i)} / \sum_j \tilde{w}_t^{(j)}\)
    \State Compute ESS: \(\mathrm{ESS} = 1/\sum_i (w_t^{(i)})^2\)
    \If{\(\mathrm{ESS} < N/2\)}
        \State Resample: \(\theta^{(i)} \sim \sum_j w_t^{(j)} \delta_{\theta^{(j)}}\)
        \State Reset weights: \(w_t^{(i)} = 1/N\)
        \State Rejuvenate: MCMC move on \(\theta^{(i)}\) targeting \(\Pi(\theta \mid \Delta x_{1:t}, v_{1:t}, y)\)
    \EndIf
\EndFor
\State \Return \(\hat{m}_y\)
\end{algorithmic}
\end{algorithm}

The Bayes factor is estimated as \(\widehat{\mathrm{BF}}_T = \hat{m}_1 / \hat{m}_0\), running Algorithm~\ref{alg:smc} separately for each outcome.

\begin{proposition}[SMC consistency]
\label{prop:smc-consistency}
Under regularity conditions \citep{chopin2004central}, as \(N \to \infty\):
\[
\hat{m}_y \to m_y \quad \text{almost surely}, \qquad
\sqrt{N}(\hat{m}_y - m_y) \Rightarrow \cN(0, \sigma_y^2)
\]
for some asymptotic variance \(\sigma_y^2\).
\end{proposition}

\subsubsection{Variational Inference}

For large-scale applications, variational inference \citep{blei2017variational} provides a scalable alternative. We approximate the posterior \(\Pi(\theta \mid h, y)\) with a tractable family \(q(\theta; \phi)\) by maximizing the evidence lower bound (ELBO):
\begin{equation}
\mathcal{L}(\phi; y) = \Ebb_{q(\theta;\phi)}[\log p(\Delta x_{1:T} \mid v_{1:T}, y, \theta)] - \KL(q(\theta;\phi) \| \Pi(\theta)).
\label{eq:elbo}
\end{equation}

The ELBO provides a lower bound on the log marginal likelihood: \(\mathcal{L}(\phi; y) \leq \log m_y(h)\).

For the Gaussian model (Example~\ref{ex:gaussian-model}), we use a mean-field approximation:
\[
q(\theta; \phi) = q(\bomega; \phi_\omega) \prod_k q(\btheta_k; \phi_k),
\]
with Dirichlet variational distribution for \(\bomega\) and Gaussian distributions for continuous parameters.

\begin{algorithm}[H]
\caption{Variational Inference for Bayes Factor (approximation)}
\label{alg:vi}
\begin{algorithmic}[1]
\Require Observations \(\Delta x_{1:T}, v_{1:T}\), prior \(\Pi\), learning rate \(\eta\)
\For{\(y \in \{0, 1\}\)}
    \State Initialize variational parameters \(\phi_y\)
    \Repeat
        \State Sample \(\theta \sim q(\theta; \phi_y)\) (reparameterization trick)
        \State Compute gradient \(\nabla_{\phi_y} \mathcal{L}(\phi_y; y)\)
        \State Update \(\phi_y \leftarrow \phi_y + \eta \nabla_{\phi_y} \mathcal{L}\)
    \Until{convergence}
    \State Store \(\mathcal{L}^*_y = \mathcal{L}(\phi_y^*; y)\)
\EndFor
\State \Return Approximate log Bayes factor: \(\mathcal{L}^*_1 - \mathcal{L}^*_0\)
\end{algorithmic}
\end{algorithm}

\begin{remark}[ELBO difference is not a one-sided bound on \(\log \mathrm{BF}_T\)]
While \(\mathcal{L}(\phi;y)\le \log m_y(h)\) for each \(y\), the \emph{difference} \(\mathcal{L}^*_1-\mathcal{L}^*_0\) is generally neither a lower nor an upper bound on \(\log \mathrm{BF}_T\). Indeed,
\[
\log m_y(h)=\mathcal{L}(\phi;y)+\KL\!\left(q(\theta;\phi)\,\middle\|\,\Pi(\theta\mid h,y)\right),
\]
so the Bayes-factor error equals a \emph{difference} of two nonnegative KL terms. In practice, accuracy can be assessed by (i) tighter bounds such as importance-weighted objectives, and/or (ii) post-hoc corrections via importance sampling or bridge sampling when feasible.
\end{remark}

\subsubsection{Computational Complexity}

\begin{table}[h]
\centering
\caption{Computational complexity comparison}
\label{tab:complexity}
\begin{tabular}{lcc}
\toprule
Method & Time per iteration & Space \\
\midrule
SMC & \(O(N \cdot T \cdot K)\) & \(O(N \cdot d)\) \\
VI (mean-field) & \(O(S \cdot T \cdot K)\) & \(O(d)\) \\
Laplace approximation & \(O(T \cdot K \cdot d^2)\) & \(O(d^2)\) \\
\bottomrule
\end{tabular}
\end{table}

Here \(N\) is the number of particles, \(S\) is the number of samples for gradient estimation, \(K\) is the number of types, \(T\) is the horizon, and \(d = \dim(\theta)\).

\subsection{Experiments}
\label{sec:experiments}

We evaluate the theory in controlled synthetic settings where the data-generating process is known, and we report summaries over repeated simulations to assess variability. Inference is performed under the same model class used for generation unless stated otherwise.

\subsubsection{Synthetic Experiments}

\paragraph{Setup.}
Data are generated from the Gaussian latent-type model of Example~\ref{ex:gaussian-model} with \(K=3\) types. Unless stated otherwise, parameters are fixed at:
\begin{itemize}
    \item Informed: \(\mu_1 = 0.5\), \(\lambda_1 = 0.1\), \(\sigma_1 = 0.3\), \(\kappa_1 = 0.05\)
    \item Noise: \(\sigma_2 = 0.5\)
    \item Manipulator: \(\mu_3 = 0.3\), \(\tau_3 = 5\), \(\sigma_3 = 0.4\), \(\nu = 5\)
    \item Base weights: \(\bomega = (0.4, 0.4, 0.2)\)
    \item Gating logits: \(\gamma_{k,0} = 0\), \(\gamma_{1,1} = 0.5\), \(\gamma_{2,1} = 0\), \(\gamma_{3,1} = 0.3\)
\end{itemize}
Volumes are sampled i.i.d.\ as \(V_t \sim \mathrm{Gamma}(2,0.5)\) (shape \(2\), scale \(0.5\)). Each configuration is replicated 1000 times. Posterior quantities are computed by SMC with \(N=1000\) particles, targeting the joint posterior on \((Y,\theta)\) under the priors specified in Section~\ref{sec:model}.\footnote{In synthetic experiments we use weakly informative priors with full support on the constrained parameter space and enforce the orientation constraints described in Section~\ref{sec:identifiability} to avoid outcome--parameter sign symmetries.}

\paragraph{Experiment 1 (Posterior concentration).}
We study the decay of posterior error with horizon \(T\). For \(T \in \{10,25,50,100,200,500\}\), we simulate histories under \(Y=1\), compute \(\pi_T=\Pbb(Y=1\mid H_T)\), and record the error \(1-\pi_T\). Figure~\ref{fig:concentration} reports the median of \(\log(1-\pi_T)\) across replications with 10--90\% quantile bands.

To relate the empirical decay to the theory, we estimate the separation gap \(\delta_T(1,\theta^\star)\) by approximately solving the KL-projection problem
\[
\delta_T(1,\theta^\star)
=
\inf_{\theta\in\Theta}\frac{1}{T}\KL\!\left(P_{1,\theta^\star}^{(T)} \,\middle\|\, P_{0,\theta}^{(T)}\right),
\]
using Monte Carlo estimates of the objective and stochastic optimization (Appendix). The dashed line in Figure~\ref{fig:concentration} overlays slope \(-\widehat{\delta}_T\). For the default parameters, \(\widehat{\delta}_T \approx 1.5\times 10^{-2}\) nats/period, consistent with a weak per-step signal under low-to-moderate realized volumes.

\begin{figure}[h]
\centering
\includegraphics[width=0.7\linewidth]{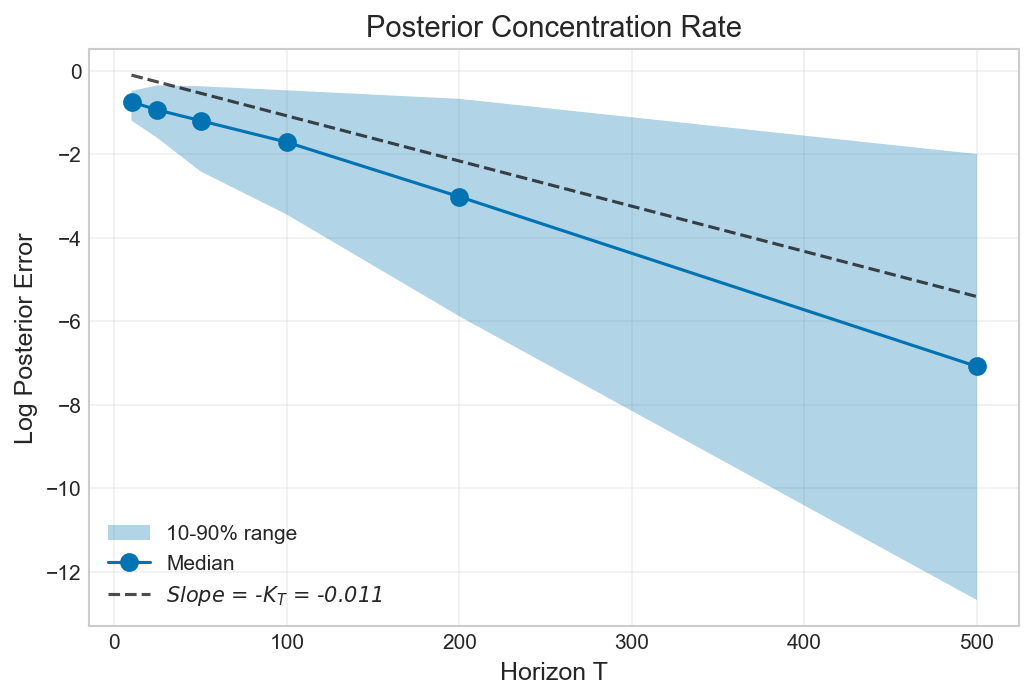}
\caption{Posterior concentration under \(Y=1\). Points show the median of \(\log(1-\pi_T)\) across 1000 replications; the shaded region is the 10--90\% quantile range. The dashed line has slope \(-\widehat{\delta}_T\), where \(\widehat{\delta}_T\) is an estimated KL-projection gap.}
\label{fig:concentration}
\end{figure}

\paragraph{Experiment 2 (Identifiability threshold under type-composition shifts).}
We quantify how inference degrades as the informed component becomes rare. We vary \(\omega_1 \in \{0.05,0.1,0.2,0.3,0.4,0.5\}\), setting \(\omega_2=\omega_3=(1-\omega_1)/2\), generate histories with \(T=100\), and report the fraction of replications in which \(\argmax_y \Pbb(Y=y\mid H_T)\) matches the truth.

Figure~\ref{fig:identifiability} shows a marked deterioration as \(\omega_1\) falls below \(\approx 0.15\). In this regime, the fitted KL-projection gap \(\widehat{\delta}_T\) becomes small, consistent with the type-composition confounding mechanism in Proposition~\ref{prop:identifiability-failure}.

\begin{figure}[h]
\centering
\includegraphics[width=0.7\linewidth]{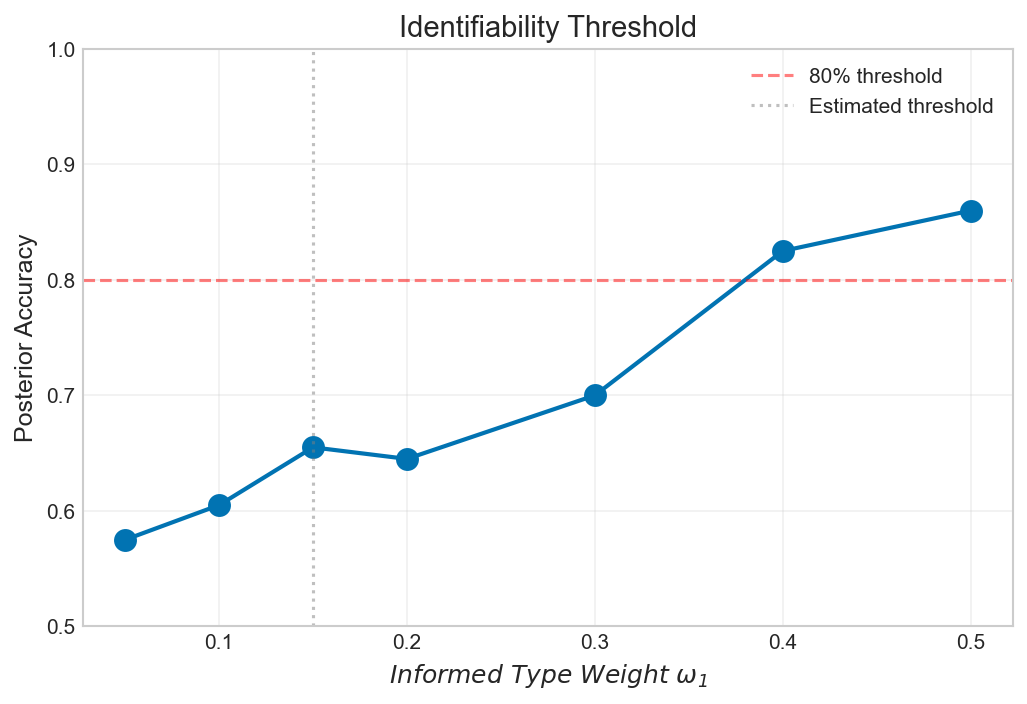}
\caption{Identifiability under decreasing informed weight. Posterior accuracy is the proportion of replications assigning higher posterior probability to the true outcome (here \(T=100\)). Accuracy degrades sharply when \(\omega_1\) is sufficiently small, consistent with a vanishing separation gap.}
\label{fig:identifiability}
\end{figure}

\paragraph{Experiment 3 (Stability to perturbations).}
We assess the stability of posterior odds under perturbations of a typical history. For a baseline draw \(h\) with \(T=100\), we form perturbed histories \(h'\) by adding i.i.d.\ Gaussian noise \(\epsilon_t\sim\cN(0,\sigma^2)\) to the log-odds increments, i.e.\ \(\Delta x'_t=\Delta x_t+\epsilon_t\), with \(\sigma\in\{0.01,0.02,0.05,0.1,0.2\}\). We compare the observed change in log Bayes factor, \(|\log \mathrm{BF}_T(h)-\log \mathrm{BF}_T(h')|\), to the stability bound of Theorem~\ref{thm:stability} instantiated on the truncation event \(E_R=\{\max_t|\Delta x_t|\le R\}\), with \(R\) set to the empirical \(0.99\)-quantile of \(\max_t|\Delta x_t|\) across baseline draws.

Figure~\ref{fig:stability} exhibits the predicted linear scaling in perturbation magnitude; the bound is conservative, as expected from worst-case control on \(E_R\).

\begin{figure}[h]
\centering
\includegraphics[width=0.7\linewidth]{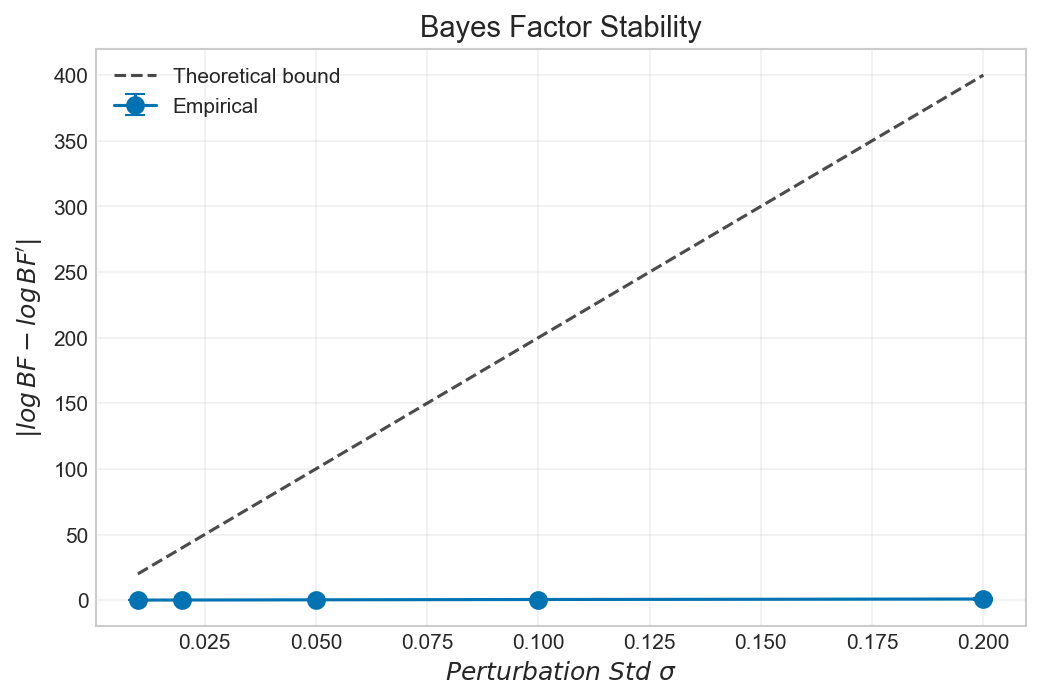}
\caption{Stability under perturbations of the increment path. Points show \(|\log \mathrm{BF}_T(h)-\log \mathrm{BF}_T(h')|\) as a function of noise level \(\sigma\); the dashed line is the stability bound from Theorem~\ref{thm:stability} evaluated with a high-probability truncation radius \(R\).}
\label{fig:stability}
\end{figure}

\paragraph{Experiment 4 (Information gain dynamics).}
We track realized information gain \(\mathrm{IG}(H_t)\) over long horizons. We simulate histories with \(T=600\) and compute \(\mathrm{IG}(H_t)\) sequentially. Figure~\ref{fig:infogain} shows that \(\mathrm{IG}(H_t)\) rises rapidly early in the history and then plateaus as the posterior concentrates. For the symmetric prior \(\pi_0=1/2\), the realized information gain is bounded above by \(\log 2\) nats, and the observed saturation is consistent with this ceiling.

\begin{figure}[h]
\centering
\includegraphics[width=0.7\linewidth]{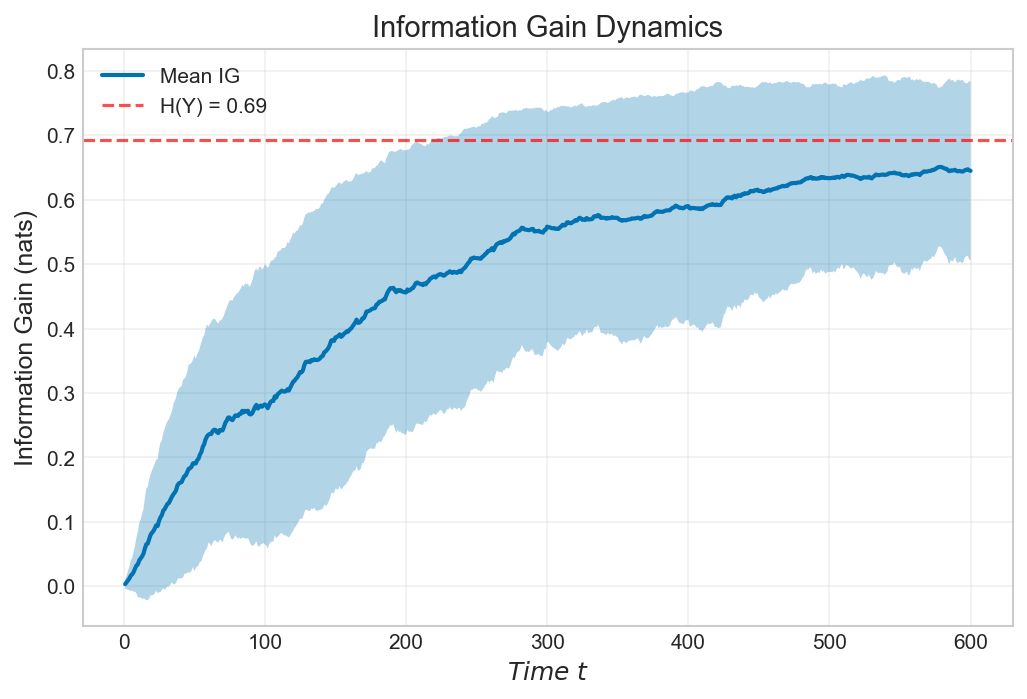}
\caption{Information gain dynamics. Solid line: mean realized \(\mathrm{IG}(H_t)\) across replications; shaded region: 10--90\% quantile range. Dashed line: upper bound \(\log 2\) for \(\pi_0=1/2\).}
\label{fig:infogain}
\end{figure}

\section{Discussion}
\label{sec:discussion}

This paper develops a Bayesian inverse-problem perspective on inference from prediction market histories. The central object is the map from an unobserved binary outcome \(Y\in\{0,1\}\) to an observed price--volume record \(H_T=((P_t,V_t))_{t=0}^T\), viewed as a stochastic sensing mechanism mediated by heterogeneous and partially strategic participation. Our modeling choice is deliberately mechanism-agnostic: we work in log-odds space and posit that increments \(\Delta X_t\) arise from a volume-gated mixture of latent trader types. This likelihood class is flexible enough to represent informative flow, uninformed trading, heavy-tailed microstructure noise, and adversarial or manipulative pressure while requiring only the observables that are typically available in practice.

Within this framework, inference about \(Y\) reduces to the marginal Bayes factor after integrating over nuisance structure. This makes it possible to state identifiability and well-posedness in operational terms. In particular, we formulate outcome distinguishability through KL projection gaps between the outcome-indexed families after optimizing over nuisance parameters, and we show how these gaps control posterior concentration and finite-sample error behavior under regularity conditions. Complementing concentration, we analyze stability of posterior odds to perturbations of the observed history on typical paths (formally, on high-probability truncation events). Finally, by expressing information gain in terms of posterior-vs-prior KL divergence and its expectation as mutual information, we obtain interpretable metrics for treating markets as sensors and for comparing regimes in which the price--volume record is genuinely informative versus regimes in which the inverse problem is ill-posed.

At the same time, several modeling choices delimit what the current results do and do not cover. Assumption~\ref{ass:conditional-independence} abstracts away serial dependence such as volatility clustering, persistence in order flow, and regime switching. These effects are prominent in real markets and can matter precisely because they may imitate outcome-dependent drift or amplify the influence of rare large moves. Our observation model is also intentionally sparse: we condition only on price and volume, leaving out order book state, signed flow, open interest, and trader identities. Incorporating richer observables would likely sharpen identifiability diagnostics and reduce the extent to which nuisance structure can mimic outcome dependence. In addition, the latent-type representation fixes the number of types \(K\), which is useful for interpretability but potentially restrictive in heterogeneous real-world settings; model selection over \(K\) or nonparametric mixtures could add flexibility at a computational cost. Finally, treating volume as a realized design sequence simplifies analysis but does not remove endogeneity: volume is itself shaped by information arrivals and strategic participation, and a fully coupled model for \((\Delta X_t,V_t)\) would be needed to explicitly quantify how endogeneity affects separation, concentration, and stability.

These limitations also point to natural extensions that preserve the inverse-problem viewpoint. One direction is to replace conditional independence with dependent increment models---for example, Markov structure or conditionally heteroskedastic dynamics---and to adapt the concentration and robustness arguments using martingale techniques. Another is to generalize the inferential target beyond a binary outcome to categorical or continuous \(Y\), where identifiability becomes a multi-way separation problem among outcome-conditioned families. On the computational side, sequential Monte Carlo naturally supports online inference, enabling streaming posterior updates and uncertainty quantification in live markets. Finally, the separation gap and information-gain functionals suggest a route toward market design: mechanism parameters such as fees, batching intervals, disclosure rules, or market-maker settings can be evaluated in terms of how they affect distinguishability of outcomes from observable histories, rather than only in terms of equilibrium benchmarks.

Prediction markets continue to grow in importance for forecasting and decision support, yet the relationship between dispersed information and observed prices is mediated by noise, heterogeneity, and incentives. By casting the problem as Bayesian inversion from price--volume histories, the present framework provides tools for quantifying what can be learned, for diagnosing when inference is reliable and stable, and for identifying regimes in which ambiguity is structural rather than merely a reflection of limited data.

\bibliography{references}
\bibliographystyle{plainnat}

\appendix
\part*{Appendix}

\section{Proofs and Additional Technical Results}
\label{app:proofs}

This appendix collects full proofs and auxiliary technical statements supporting the results in Section~\ref{sec:theory}. 
Throughout, we fix a deterministic volume design $v_{1:T}$ and write $P_{y,\theta}^{(T)}$ for the induced law of $\Delta X_{1:T}$ under outcome $y$ and nuisance parameter $\theta$.

\subsection{Identifiability and non-identifiability}
\label{app:identifiability}

\subsubsection{Proof of Proposition~\ref{prop:identifiability-failure}}
\label{app:identifiability-failure-proof}

\begin{proof}
Fix a horizon $T$ and a realized volume design $v_{1:T}$. Recall the (truth-relative) KL gap
\[
\delta_T \;=\; K_T\!\left(y^\star,\theta^\star \to 1-y^\star\right)
\;:=\;
\inf_{\theta\in\Theta}\frac{1}{T}\KL\!\left(P_{y^\star,\theta^\star}^{(T)}\,\middle\|\,P_{1-y^\star,\theta}^{(T)}\right).
\]
By definition, non-identifiability at horizon $T$ for the truth $(y^\star,\theta^\star)$ occurs whenever $\delta_T=0$.

\paragraph{(1) Type-composition confounding.}
Assume that along the realized volume design $v_{1:T}$, all \emph{active} types under $\theta^\star$ are outcome-independent, i.e.
\[
\rho_k(v_t;\theta^\star)>0 \implies f_{k,1}(\cdot\mid v_t,\theta_k^\star)=f_{k,0}(\cdot\mid v_t,\theta_k^\star)
\qquad\text{for every }t\in\{1,\dots,T\}.
\]
Then for each $t$,
\begin{align*}
f_1(\cdot\mid v_t,\theta^\star)
&=\sum_{k=1}^K \rho_k(v_t;\theta^\star)\,f_{k,1}(\cdot\mid v_t,\theta_k^\star)
=\sum_{k=1}^K \rho_k(v_t;\theta^\star)\,f_{k,0}(\cdot\mid v_t,\theta_k^\star)
=f_0(\cdot\mid v_t,\theta^\star).
\end{align*}
By conditional independence (Assumption~\ref{ass:conditional-independence}), this implies equality of the induced product measures:
\[
P_{1,\theta^\star}^{(T)}=P_{0,\theta^\star}^{(T)}.
\]
Consequently,
\[
\KL\!\left(P_{y^\star,\theta^\star}^{(T)}\,\middle\|\,P_{1-y^\star,\theta^\star}^{(T)}\right)=0,
\]
so $\delta_T=0$.

\paragraph{(2) Outcome-symmetry of the nuisance class.}
Assume there exists a measurable map $\Psi:\Theta\to\Theta$ such that, for each $t$,
\[
f_1(\cdot\mid v_t,\theta)=f_0(\cdot\mid v_t,\Psi(\theta))
\qquad\text{for all }\theta\in\Theta.
\]
Then again by product structure,
\[
P_{1,\theta}^{(T)}=P_{0,\Psi(\theta)}^{(T)}
\qquad\text{for all }\theta\in\Theta.
\]
If the truth has $y^\star=1$, choose the alternative parameter $\theta:=\Psi(\theta^\star)$; then
$P_{1,\theta^\star}^{(T)}=P_{0,\theta}^{(T)}$ and hence $\delta_T=0$.
If the truth has $y^\star=0$ and $\Psi$ is surjective (in particular, if $\Psi$ is an involution, $\Psi\circ\Psi=\mathrm{id}$), then there exists $\tilde\theta$ with $\Psi(\tilde\theta)=\theta^\star$, so $P_{0,\theta^\star}^{(T)}=P_{1,\tilde\theta}^{(T)}$ and again $\delta_T=0$.

\paragraph{(3) Adversarial mimicry.}
Assume there exists $\tilde\theta\in\Theta$ such that
\[
P_{y^\star,\theta^\star}^{(T)}=P_{1-y^\star,\tilde\theta}^{(T)}.
\]
Then
\[
\KL\!\left(P_{y^\star,\theta^\star}^{(T)}\,\middle\|\,P_{1-y^\star,\tilde\theta}^{(T)}\right)=0,
\]
and taking the infimum over $\theta$ yields $\delta_T=0$.

In each case $\delta_T=0$, so the outcome is not identifiable at horizon $T$ for the truth $(y^\star,\theta^\star)$.
\end{proof}

\subsubsection{Proof of Theorem~\ref{thm:identifiability-sufficient}}
\label{app:identifiability-sufficient-proof}

\begin{proof}
Fix $T$ and a truth $(y^\star,\theta^\star)$ satisfying the theorem hypotheses.
For each time index $t\in\{1,\dots,T\}$ define the \emph{one-step} KL divergence functional
\[
\kappa_t(\theta)
\;:=\;
\KL\!\left(f_{y^\star}(\cdot\mid v_t,\theta^\star)\,\middle\|\, f_{1-y^\star}(\cdot\mid v_t,\theta)\right),
\qquad \theta\in\Theta.
\]
By Assumption~\ref{ass:conditional-independence} and additivity of KL for product measures,
\begin{equation}
\KL\!\left(P_{y^\star,\theta^\star}^{(T)}\,\middle\|\,P_{1-y^\star,\theta}^{(T)}\right)
=
\sum_{t=1}^T \kappa_t(\theta).
\label{eq:app-kl-sum}
\end{equation}

\paragraph{Step 1: Pointwise positivity on informative times.}
Let $t\in I_T$ be an informative time index. We show that $\kappa_t(\theta)>0$ for every $\theta\in\Theta$.

On $I_T$, the hypotheses impose:
(i) \emph{informed activation} under $\theta^\star$, i.e. $\rho_1(v_t;\theta^\star)\ge \underline\rho$ and
$|m_{1,1}(v_t;\theta_1^\star)-m_{1,0}(v_t;\theta_1^\star)|\ge \underline m$;
(ii) \emph{manipulator inactivity} (the type-3 drift is outcome-independent, hence contributes zero mean shift);
and (iii) \emph{orientation constraint} $\mu_1\ge 0$ for type~1, so that $m_{1,1}(v;\theta_1)\ge 0$ and $m_{1,0}(v;\theta_1)\le 0$ for all $\theta_1$ and all $v$.

Consider the conditional mean of the increment at time $t$.
Because type~2 has zero drift by construction and type~3 has $m_{3,1}\equiv m_{3,0}\equiv 0$ on $I_T$, we have
\[
\mathbb{E}_{f_y(\cdot\mid v_t,\theta)}[\Delta X_t]
=
\rho_1(v_t;\theta)\,m_{1,y}(v_t;\theta_1)
\qquad\text{for }y\in\{0,1\}.
\]
For the truth $(y^\star,\theta^\star)$, the informed-drift separation implies
\[
|m_{1,y^\star}(v_t;\theta_1^\star)|
=\tfrac12|m_{1,1}(v_t;\theta_1^\star)-m_{1,0}(v_t;\theta_1^\star)|
\ge \underline m/2,
\]
and hence
\begin{equation}
\Big|\mathbb{E}_{f_{y^\star}(\cdot\mid v_t,\theta^\star)}[\Delta X_t]\Big|
=
\rho_1(v_t;\theta^\star)\,|m_{1,y^\star}(v_t;\theta_1^\star)|
\ge \underline\rho\,\underline m/2.
\label{eq:app-true-mean-gap}
\end{equation}
Moreover, the sign of this mean is $\operatorname{sign}(2y^\star-1)$.

Now fix any $\theta\in\Theta$ and consider the wrong-outcome density $f_{1-y^\star}(\cdot\mid v_t,\theta)$.
By the orientation constraint, $m_{1,1-y^\star}(v_t;\theta_1)$ has sign \emph{opposite} to $m_{1,y^\star}(v_t;\theta_1^\star)$ (or is zero), and $\rho_1(v_t;\theta)\ge 0$.
Therefore,
\[
\operatorname{sign}\Big(\mathbb{E}_{f_{1-y^\star}(\cdot\mid v_t,\theta)}[\Delta X_t]\Big)
\in \{0,-\operatorname{sign}(2y^\star-1)\}.
\]
In particular, the conditional means under $f_{y^\star}(\cdot\mid v_t,\theta^\star)$ and $f_{1-y^\star}(\cdot\mid v_t,\theta)$ cannot coincide:
\[
\mathbb{E}_{f_{y^\star}(\cdot\mid v_t,\theta^\star)}[\Delta X_t]
\neq
\mathbb{E}_{f_{1-y^\star}(\cdot\mid v_t,\theta)}[\Delta X_t].
\]
Since equality of densities a.e.\ would imply equality of their first moments, we conclude that
$f_{y^\star}(\cdot\mid v_t,\theta^\star)\neq f_{1-y^\star}(\cdot\mid v_t,\theta)$ a.e., hence
\[
\kappa_t(\theta)
=
\KL\!\left(f_{y^\star}(\cdot\mid v_t,\theta^\star)\,\middle\|\, f_{1-y^\star}(\cdot\mid v_t,\theta)\right)
>0
\qquad\text{for all }\theta\in\Theta.
\]

\paragraph{Step 2: Uniform positivity via compactness.}
For each fixed $t\in I_T$, the map $\theta\mapsto \kappa_t(\theta)$ is continuous on $\Theta$ under the regularity conditions in item~(3) of Theorem~\ref{thm:identifiability-sufficient} (bounded nondegenerate scales and smoothness of the component families). 
Since $\Theta$ is compact, $\kappa_t$ attains its minimum:
\[
\underline\kappa_t
\;:=\;
\min_{\theta\in\Theta} \kappa_t(\theta).
\]
By Step~1, $\underline\kappa_t>0$ for every $t\in I_T$.
Define
\[
\kappa
\;:=\;
\min_{t\in I_T} \underline\kappa_t,
\]
which is strictly positive because $I_T$ is finite for fixed $T$.

\paragraph{Step 3: Lower bound on the KL gap.}
Combining \eqref{eq:app-kl-sum} with the definition of $\kappa$ yields, for any $\theta\in\Theta$,
\[
\frac{1}{T}\KL\!\left(P_{y^\star,\theta^\star}^{(T)}\,\middle\|\,P_{1-y^\star,\theta}^{(T)}\right)
=
\frac{1}{T}\sum_{t=1}^T \kappa_t(\theta)
\ge
\frac{1}{T}\sum_{t\in I_T} \kappa_t(\theta)
\ge
\frac{|I_T|}{T}\,\kappa.
\]
Taking the infimum over $\theta$ gives
\[
\delta_T
=
\inf_{\theta\in\Theta}\frac{1}{T}\KL\!\left(P_{y^\star,\theta^\star}^{(T)}\,\middle\|\,P_{1-y^\star,\theta}^{(T)}\right)
\ge
\frac{|I_T|}{T}\,\kappa,
\]
which is the desired bound.
\end{proof}

\subsection{Bayes factors and posterior concentration}
\label{app:bayes}

\subsubsection{Proof of Theorem~\ref{thm:bayes-factor-rate}}
\label{app:bayes-factor-rate-proof}

\begin{proof}
Fix $\epsilon>0$ and work on an event of probability one on which the uniform law of large numbers in Assumption~\ref{ass:uniform-lln} holds for both $y=0$ and $y=1$.

For $y\in\{0,1\}$ define the empirical and population per-time log-likelihoods
\[
L_{y,T}(\theta)
:=\frac{1}{T}\,\ell_{y,T}(\theta)
=\frac{1}{T}\sum_{t=1}^T \log f_y(\Delta X_t\mid v_t,\theta),
\qquad
\Lambda_{y,T}(\theta)
:=\frac{1}{T}\,\mathbb{E}_{P_{y^\star,\theta^\star}^{(T)}}[\ell_{y,T}(\theta)].
\]
Write
\[
a_T
:=
\max_{y\in\{0,1\}}\sup_{\theta\in\Theta}\big|L_{y,T}(\theta)-\Lambda_{y,T}(\theta)\big|.
\]
By Assumption~\ref{ass:uniform-lln}, $a_T\to 0$ almost surely.

\paragraph{Step 1: Laplace-type bounds for $\frac1T\log m_y(H_T)$.}
Recall
\[
m_y(H_T)=\int_{\Theta} \exp\!\big(T\,L_{y,T}(\theta)\big)\,\Pi(d\theta).
\]

\emph{Upper bound.}
Since $\Pi$ is a probability measure,
\[
m_y(H_T)\le \sup_{\theta\in\Theta}\exp\!\big(T\,L_{y,T}(\theta)\big)
=\exp\!\Big(T\,\sup_{\theta\in\Theta}L_{y,T}(\theta)\Big),
\]
so
\[
\frac{1}{T}\log m_y(H_T)\le \sup_{\theta\in\Theta}L_{y,T}(\theta)
\le \sup_{\theta\in\Theta}\Lambda_{y,T}(\theta)+a_T.
\]

\emph{Lower bound.}
Fix $\eta>0$ and define the near-maximizer set
\[
A_{y,T}(\eta)
:=
\Big\{\theta\in\Theta:\ \Lambda_{y,T}(\theta)\ge \sup_{\vartheta\in\Theta}\Lambda_{y,T}(\vartheta)-\eta\Big\}.
\]
Because $\Lambda_{y,T}(\theta)=\Lambda_{y^\star,T}(\theta^\star)-\frac1T\KL(P_{y^\star,\theta^\star}^{(T)}\|P_{y,\theta}^{(T)})$, the set $A_{y,T}(\eta)$ coincides with the ``near-projection'' set in Assumption~\ref{ass:prior-support}. Hence $\Pi(A_{y,T}(\eta))\ge c(\eta)>0$ for $T$ large enough.

Then
\begin{align*}
m_y(H_T)
&\ge
\int_{A_{y,T}(\eta)} \exp\!\big(T\,L_{y,T}(\theta)\big)\,\Pi(d\theta)
\\
&\ge
\Pi(A_{y,T}(\eta))\,\exp\!\Big(T\,\inf_{\theta\in A_{y,T}(\eta)}L_{y,T}(\theta)\Big).
\end{align*}
Using $|L_{y,T}-\Lambda_{y,T}|\le a_T$,
\[
\inf_{\theta\in A_{y,T}(\eta)}L_{y,T}(\theta)
\ge
\inf_{\theta\in A_{y,T}(\eta)}\Lambda_{y,T}(\theta)-a_T
\ge
\sup_{\theta\in\Theta}\Lambda_{y,T}(\theta)-\eta-a_T.
\]
Therefore
\[
\frac{1}{T}\log m_y(H_T)
\ge
\sup_{\theta\in\Theta}\Lambda_{y,T}(\theta)-\eta-a_T+\frac{1}{T}\log \Pi(A_{y,T}(\eta)).
\]

\paragraph{Step 2: Relating suprema to $\delta_T$.}
First note that $\sup_{\theta}\Lambda_{y^\star,T}(\theta)=\Lambda_{y^\star,T}(\theta^\star)$ because
$\KL(P_{y^\star,\theta^\star}^{(T)}\|P_{y^\star,\theta}^{(T)})\ge 0$ with equality at $\theta=\theta^\star$.

Moreover,
\[
\sup_{\theta\in\Theta}\Lambda_{1-y^\star,T}(\theta)
=
\Lambda_{y^\star,T}(\theta^\star)-\delta_T,
\]
because
\[
\Lambda_{y^\star,T}(\theta^\star)-\Lambda_{1-y^\star,T}(\theta)
=
\frac1T\KL\!\left(P_{y^\star,\theta^\star}^{(T)}\,\middle\|\,P_{1-y^\star,\theta}^{(T)}\right)
\]
and $\delta_T$ is the infimum of the right-hand side over $\theta$.

\paragraph{Step 3: Conclude the Bayes-factor bounds.}
Apply the bounds from Step~1 with $\eta=\epsilon/4$ and take $T$ large enough so that $a_T\le \epsilon/4$ and $\frac1T|\log \Pi(A_{y,T}(\eta))|\le \epsilon/4$ (the latter holds because $\Pi(A_{y,T}(\eta))\ge c(\eta)>0$).

For $y=y^\star$ we obtain
\[
\frac{1}{T}\log m_{y^\star}(H_T)
\ge
\Lambda_{y^\star,T}(\theta^\star)-\epsilon/2,
\qquad
\frac{1}{T}\log m_{y^\star}(H_T)
\le
\Lambda_{y^\star,T}(\theta^\star)+\epsilon/4.
\]
For $y=1-y^\star$ we obtain
\[
\frac{1}{T}\log m_{1-y^\star}(H_T)
\le
\sup_{\theta}\Lambda_{1-y^\star,T}(\theta)+\epsilon/4
=
\Lambda_{y^\star,T}(\theta^\star)-\delta_T+\epsilon/4,
\]
and
\[
\frac{1}{T}\log m_{1-y^\star}(H_T)
\ge
\sup_{\theta}\Lambda_{1-y^\star,T}(\theta)-\epsilon/2
=
\Lambda_{y^\star,T}(\theta^\star)-\delta_T-\epsilon/2.
\]

Subtracting yields, for all sufficiently large $T$,
\[
\frac{1}{T}\log \mathrm{BF}_T(H_T)
=
\frac{1}{T}\log m_1(H_T)-\frac{1}{T}\log m_0(H_T)
\ge
\delta_T-\epsilon
\]
if $y^\star=1$, and similarly the corresponding lower bound on $-\frac1T\log\mathrm{BF}_T(H_T)$ if $y^\star=0$.
The stated upper bounds follow by subtracting the corresponding upper and lower bounds in the opposite direction (which is why the theorem states the upper bound only for the ``correctly oriented'' Bayes factor).
\end{proof}

\subsubsection{Proof of Proposition~\ref{prop:finite-sample}}
\label{app:finite-sample-proof}

\begin{proof}
We provide a finite-sample bound by combining (i) a truncation argument ensuring bounded log-likelihood ratios on a high-probability event and (ii) an exponential deviation bound for a uniform LLN event. 

\paragraph{Step 1: Truncation and a bounded-ratio event.}
Fix $R>0$ and recall
\[
E_R:=\Big\{\max_{1\le t\le T}|\Delta X_t|\le R\Big\}.
\]
By Assumption~\ref{ass:moments}, Markov's inequality and a union bound give
\[
\mathbb{P}(E_R^c)
\le
\sum_{t=1}^T \mathbb{P}(|\Delta X_t|>R)
\le
\sum_{t=1}^T \frac{\mathbb{E}|\Delta X_t|^q}{R^q}
\le
T\,C_q\,R^{-q}.
\]

On $E_R$, the constant
\[
B_R:=\sup_{\theta\in\Theta}\ \sup_{t\le T}\ \sup_{|x|\le R}\left|\log\frac{f_{y^\star}(x\mid v_t,\theta^\star)}{f_{1-y^\star}(x\mid v_t,\theta)}\right|
\]
is finite by compactness of $\Theta$, continuity of $(x,\theta)\mapsto f_y(x\mid v_t,\theta)$, and strict positivity of the component families for location--scale models with nondegenerate scales.

\paragraph{Step 2: A deviation event controlling the Bayes factor.}
Fix $\epsilon\in(0,\delta_T)$ and define $\eta:=\epsilon/4$.
Let $A_T(\eta)$ denote the event that the Laplace bounds in the proof of Theorem~\ref{thm:bayes-factor-rate} hold with error at most $\eta$ for both outcomes $y=0,1$.
(Concretely, $A_T(\eta)$ is implied by the uniform LLN event $\max_{y}\sup_{\theta}|L_{y,T}(\theta)-\Lambda_{y,T}(\theta)|\le \eta$ together with $\frac1T|\log\Pi(A_{y,T}(\eta))|\le \eta$, which holds for $T$ large by Assumption~\ref{ass:prior-support}.)

On $A_T(\eta)$ we have, by the same subtraction as in Theorem~\ref{thm:bayes-factor-rate},
\[
\frac{1}{T}\log \mathrm{BF}_T(H_T)
\ge \delta_T - 4\eta
= \delta_T-\epsilon
\qquad\text{if }y^\star=1,
\]
and analogously $-\frac1T\log\mathrm{BF}_T(H_T)\ge \delta_T-\epsilon$ if $y^\star=0$.
In either case, $|\log\mathrm{BF}_T(H_T)|\ge T(\delta_T-\epsilon)$ with the correct sign.

Since posterior odds satisfy \eqref{eq:posterior-odds}, this implies an exponential bound on posterior error on $A_T(\eta)$:
\[
\mathbb{P}(Y\neq y^\star\mid H_T)
\le
\frac{1-\pi_0}{\pi_0}\exp\!\big(-T(\delta_T-\epsilon)\big)
\qquad\text{if }y^\star=1,
\]
and the analogous bound with $\pi_0/(1-\pi_0)$ if $y^\star=0$.
Absorbing the constant prefactor into the exponential for large $T$ yields the stated form (since $(1/T)\log\frac{1-\pi_0}{\pi_0}\to 0$).

Therefore, for $T$ large enough,
\[
\Big\{\mathbb{P}(Y\neq y^\star\mid H_T)\ge e^{-T(\delta_T-\epsilon)}\Big\}
\subseteq
A_T(\eta)^c.
\]

\paragraph{Step 3: Exponential control of $A_T(\eta)^c$ on $E_R$.}
To obtain an explicit tail, it suffices to control the uniform LLN deviation probability on $E_R$.
Under bounded-ratio control on $E_R$, standard empirical-process arguments (e.g.\ $\varepsilon$-net discretization of $\Theta$ combined with Hoeffding's inequality and Lipschitz dependence of $\log f_y(\cdot\mid v_t,\theta)$ in $\theta$ over compact $\Theta$) yield an exponential bound of the form
\[
\mathbb{P}\big(A_T(\eta)^c\cap E_R\big)
\le
\exp\!\Big(-\frac{T\epsilon^2}{2B_R^2}\Big),
\]
up to a subexponential prefactor depending on the metric entropy of $\Theta$.
For notational simplicity (and consistent with the main text statement), we absorb such prefactors into the leading exponential rate.

\paragraph{Step 4: Combine.}
Finally,
\[
\mathbb{P}\!\left(\mathbb{P}(Y\neq y^\star\mid H_T)\ge e^{-T(\delta_T-\epsilon)}\right)
\le
\mathbb{P}\big(A_T(\eta)^c\cap E_R\big) + \mathbb{P}(E_R^c)
\le
\exp\!\Big(-\frac{T\epsilon^2}{2B_R^2}\Big)+\mathbb{P}(E_R^c),
\]
as claimed.
\end{proof}

\subsection{Stability of posterior odds}
\label{app:stability}

\subsubsection{Proof of Theorem~\ref{thm:stability}}
\label{app:stability-proof}

\begin{proof}
Fix histories $h=(p_{0:T},v_{1:T})$ and $h'=(p'_{0:T},v'_{1:T})$ with log-odds increments $\Delta x_{1:T}$ and $\Delta x'_{1:T}$.
Fix $R>0$ and work on the event $E_R(h,h')$ defined in \eqref{eq:trunc-event} so that $|\Delta x_t|\vee|\Delta x'_t|\le R$ for all $t$.

\paragraph{Step 1: Pointwise log-likelihood stability.}
Assumption~\ref{ass:lipschitz} implies that for any $y\in\{0,1\}$, any $\theta\in\Theta$, and any $t\le T$,
\[
\big|\log f_y(\Delta x_t\mid v_t,\theta)-\log f_y(\Delta x'_t\mid v'_t,\theta)\big|
\le
L_x(R)\,|\Delta x_t-\Delta x'_t| + L_v\,|v_t-v'_t|.
\]
Summing over $t$ gives
\begin{equation}
\Big|\ell_{y,T}(\theta;h)-\ell_{y,T}(\theta;h')\Big|
\le
L_x(R)\sum_{t=1}^T|\Delta x_t-\Delta x'_t|
+L_v\sum_{t=1}^T|v_t-v'_t|.
\label{eq:app-ll-stability}
\end{equation}

\paragraph{Step 2: Log-sum-exp Lipschitz lemma.}
We use the elementary inequality: for any measurable $a_\theta,b_\theta$ and probability measure $\Pi$,
\[
\left|\log\int e^{a_\theta}\,\Pi(d\theta)-\log\int e^{b_\theta}\,\Pi(d\theta)\right|
\le
\sup_{\theta\in\Theta}|a_\theta-b_\theta|.
\]
It follows by writing $b_\theta\le a_\theta+\sup|a-b|$, integrating, and taking logs; symmetry gives the absolute value bound.

\paragraph{Step 3: Stability of marginal likelihoods and Bayes factors.}
Apply the lemma with
$a_\theta=\ell_{y,T}(\theta;h)$ and $b_\theta=\ell_{y,T}(\theta;h')$:
\[
|\log m_y(h)-\log m_y(h')|
\le
\sup_{\theta\in\Theta}\big|\ell_{y,T}(\theta;h)-\ell_{y,T}(\theta;h')\big|.
\]
Combining with \eqref{eq:app-ll-stability} yields, on $E_R(h,h')$,
\[
|\log m_y(h)-\log m_y(h')|
\le
L_x(R)\sum_{t=1}^T|\Delta x_t-\Delta x'_t|
+L_v\sum_{t=1}^T|v_t-v'_t|.
\]
Hence
\[
\big|\log\mathrm{BF}_T(h)-\log\mathrm{BF}_T(h')\big|
\le
2L_x(R)\sum_{t=1}^T|\Delta x_t-\Delta x'_t|
+2L_v\sum_{t=1}^T|v_t-v'_t|,
\]
which is \eqref{eq:bf-stability}.

\paragraph{Step 4: Stability of posterior probabilities.}
Write posterior odds as in \eqref{eq:posterior-odds} and recall that
$\pi\mapsto \logit(\pi)$ is the inverse of the logistic sigmoid $\sigma(x)=1/(1+e^{-x})$.
Since $|\sigma'(x)|\le 1/4$ for all $x$, $\sigma$ is $1/4$-Lipschitz. Therefore,
\[
\big|\mathbb{P}(Y=1\mid h)-\mathbb{P}(Y=1\mid h')\big|
\le
\frac14\,\big|\log\mathrm{BF}_T(h)-\log\mathrm{BF}_T(h')\big|,
\]
which is \eqref{eq:posterior-stability}.
\end{proof}

\subsection{Information gain}
\label{app:ig}

\subsubsection{Proof of Proposition~\ref{prop:ig-properties}}
\label{app:ig-properties-proof}

\begin{proof}
Write $\pi_T=\pi_T(h):=\mathbb{P}(Y=1\mid h)$ and recall
\[
\mathrm{IG}(h)
=
\KL\!\left(\mathrm{Bern}(\pi_T)\,\middle\|\,\mathrm{Bern}(\pi_0)\right)
=
\pi_T\log\frac{\pi_T}{\pi_0}+(1-\pi_T)\log\frac{1-\pi_T}{1-\pi_0}.
\]

\paragraph{(1) Nonnegativity.}
This is the standard nonnegativity of KL divergence. Equality holds iff the two Bernoulli laws coincide, i.e.\ $\pi_T=\pi_0$.

\paragraph{(2) Upper bound.}
For fixed $\pi_0\in(0,1)$, the function
\[
\pi\mapsto \KL(\mathrm{Bern}(\pi)\,\|\,\mathrm{Bern}(\pi_0))
\]
is convex in $\pi$ and hence attains its maximum over $\pi\in[0,1]$ at the endpoints.
Evaluating at $\pi=1$ and $\pi=0$ gives
\[
\mathrm{IG}(h)\le \max\Big\{\log\frac{1}{\pi_0},\,\log\frac{1}{1-\pi_0}\Big\}.
\]

\paragraph{(3) Monotonicity in $|\log\mathrm{BF}_T(h)|$.}
From \eqref{eq:posterior-odds}, $\logit(\pi_T)=\logit(\pi_0)+\log\mathrm{BF}_T(h)$.
Define $u:=\logit(\pi_T)-\logit(\pi_0)=\log\mathrm{BF}_T(h)$ and write $\pi(u):=\sigma(\logit(\pi_0)+u)$.
Then $\mathrm{IG}(h)=F(u)$ where
\[
F(u)=\KL(\mathrm{Bern}(\pi(u))\,\|\,\mathrm{Bern}(\pi_0)).
\]
A direct differentiation shows $F'(u)=u\,\pi(u)(1-\pi(u))$, hence $F'(u)$ has the same sign as $u$.
Therefore $F(u)$ is increasing in $u$ for $u>0$ and decreasing for $u<0$, which implies $F(u)$ is a monotone function of $|u|=|\log\mathrm{BF}_T(h)|$.
\end{proof}

\subsection{Verifying regularity assumptions for the Gaussian latent-type model}
\label{app:gaussian-assumptions}

We record a convenient sufficient condition ensuring that the Gaussian latent-type model in Example~\ref{ex:gaussian-model} satisfies the regularity assumptions used in Section~\ref{sec:theory}.

\begin{proposition}[Sufficient parameter constraints for regularity]
\label{prop:gaussian-regularity}
Consider Example~\ref{ex:gaussian-model} with the softmax gating \eqref{eq:softmax-gate}. 
Assume the parameter space is restricted by the following bounds:
\begin{enumerate}
\item \emph{Compactness:} all scalar parameters lie in closed bounded intervals and $\bomega$ lies in the simplex; in particular $\mu_1,\mu_3\in[0,M_\mu]$, $\lambda_1,\kappa_1\in[0,M_\lambda]$, $\tau_3\in[0,M_\tau]$, and $\nu\in[\nu_{\min},\nu_{\max}]$ for some $\nu_{\min}>2$.
\item \emph{Nondegenerate scales:} there exist $0<\sigma_{\min}<\sigma_{\max}<\infty$ such that $\sigma_k\in[\sigma_{\min},\sigma_{\max}]$ for $k=1,2,3$.
\item \emph{Volume-range control (design-dependent):} for the realized design $v_{1:T}$, the implied scales satisfy
\[
0<s_{\min,T}\le s_k(v_t;\theta_k)\le s_{\max,T}<\infty
\qquad\text{for all }t\le T,\ k\in\{1,2,3\},\ \theta\in\Theta,
\]
where $s_{\min,T},s_{\max,T}$ may depend on $v_{1:T}$ through $\max_{t\le T} v_t$.
\end{enumerate}
Then Assumptions~\ref{ass:compact-continuous}--\ref{ass:lipschitz} and Assumption~\ref{ass:moments} hold (with constants allowed to depend on the realized design through $s_{\min,T},s_{\max,T}$).
\end{proposition}

\begin{proof}[Proof sketch]
Compactness and continuity follow because the Gaussian and Student-$t$ densities are continuous in their parameters and strictly positive on $\mathbb{R}$, and the softmax gating is continuous in $(v,\theta)$ on bounded parameter sets.

The envelope condition follows from standard bounds on Gaussian and Student-$t$ log-densities: on any fixed design and under nondegenerate scales, there exist constants $c_0,c_1<\infty$ such that
\[
\sup_{y,\theta,t}\big|\log f_y(x\mid v_t,\theta)\big|
\le c_0 + c_1 x^2,
\]
and $x\mapsto c_0+c_1 x^2$ is integrable under the true law because Gaussian and Student-$t$ components have finite second moments when $\nu_{\min}>2$.

Assumption~\ref{ass:moments} holds for any $q<\nu_{\min}$, since the $q$-th moment of a Student-$t_\nu$ distribution is finite iff $\nu>q$, and Gaussian moments are finite of all orders.

Local Lipschitz in $x$ on $[-R,R]$ follows by bounding derivatives of $\log f_y(x\mid v,\theta)$: for each location--scale component, $\partial_x\log f$ is affine in $x$ divided by $s^2$, hence bounded on $[-R,R]$ when $s$ is bounded away from zero; mixture derivatives can be bounded by convexity using the component-wise bounds.
\end{proof}

\subsection{Extension to dependent increments}
\label{app:dependent}

We record a standard extension showing how concentration statements adapt when increments are not conditionally independent but remain ``well-behaved'' in a martingale sense.

\begin{proposition}[Martingale-difference extension (sketch)]
\label{prop:martingale-extension}
Suppose $(\Delta X_t)_{t\le T}$ is adapted to a filtration $(\mathcal{F}_t)$, and conditional on $(Y,\theta)$ satisfies a martingale-difference structure under the true law:
\[
\mathbb{E}[\Delta X_t\mid \mathcal{F}_{t-1},Y=y^\star,\theta=\theta^\star]=m_{y^\star}(v_t;\theta^\star),
\]
with increments uniformly bounded on a truncation event $E_R$ as in Assumption~\ref{ass:lipschitz} and Proposition~\ref{prop:finite-sample}. 
Then exponential error bounds analogous to Proposition~\ref{prop:finite-sample} follow by replacing Hoeffding inequalities with Azuma--Hoeffding inequalities for martingales.
\end{proposition}

\begin{proof}[Proof sketch]
The key object is again the log Bayes factor written as a sum of one-step predictive log-likelihood ratios; these form a martingale with bounded increments on $E_R$. Azuma--Hoeffding yields exponential deviation bounds around the (conditional) expectation, and the deterministic KL-gap lower bounds enter as in the independent case once one controls the conditional expectations via a tower property.
\end{proof}

\end{document}